\documentclass[11pt]{article}

\usepackage[utf8]{inputenc}
\usepackage[english]{babel}
\usepackage[T1]{fontenc}

\usepackage[hidelinks]{hyperref}
\usepackage{url}

\usepackage{graphicx}
\usepackage{algorithm}
\usepackage{placeins}
\usepackage{algpseudocode}
\usepackage[skip=2pt]{caption}

\usepackage{mathtools}
\usepackage{amsmath, amsthm, amssymb,amsfonts}
\usepackage{thmtools,thm-restate}
\usepackage{dsfont}
\usepackage{physics}

\usepackage{booktabs}

\usepackage{natbib}
\usepackage{breakcites}

\usepackage[a4paper, total={16cm, 22cm}]{geometry}

\newtheorem{theorem}{Theorem}

\newtheorem{remark}[theorem]{Remark}

\newtheorem{lemma}[theorem]{Lemma}

\newcommand{\BE}{{\mathbb{E}}}

\newcommand{\BP}{{\mathbb{P}}}
\newcommand{\BQ}{{\mathbb{Q}}}
\newcommand{\BR}{{\mathbb{R}}}

\newcommand{\CA}{{\cal A}}

\newcommand{\CO}{{\cal O}}

\newcommand{\bX}{{\bf X}}

\newcommand{\by}{{\bf y}}

\DeclareMathOperator{\Exp}{Exp}

\DeclareMathOperator*{\argmax}{argmax}

\DeclareMathSymbol{\shortminus}{\mathbin}{AMSa}{"39}

\def\indicator{\mathds{1}}

\def\uv{{(u,v]}}
\def\us{{(u,s]}}
\def\sv{{(s,v]}}
\newcommand{\seg}[2]{{(#1, #2]}}
\def\segk{{\seg{\alpha_{k \shortminus 1}}{\alpha_k}}}
\def\segkzero{{\seg{\alpha_{k \shortminus 1}^0}{\alpha_k^0}}}
\def\segj{{\seg{\alpha_{j \shortminus 1}}{\alpha_j}}}

\DeclarePairedDelimiterX{\infdivx}[2]{(}{)}{%
  #1\;\delimsize\|\;#2%
}
\newcommand{\DKL}{D_{KL}\infdivx}

\newcommand{\dTV}[2]{\dd_{\mathrm{TV}}(#1, #2)}

\let\phi\varphi
\let\theta\vartheta
\let\epsilon\varepsilon

\let\emptyset\varnothing

\let\leq\leqslant
\let\geq\geqslant

\algdef{SE}[SUBALG]{Indent}{EndIndent}{}{\algorithmicend\ }%
\algtext*{Indent}
\algtext*{EndIndent}

\long\def\contrib#1{\vskip 0.3in\noindent{\large\bf Authors’ contributions}\vskip 0.2in
\noindent #1}

\long\def\acks#1{\vskip 0.3in\noindent{\large\bf Acknowledgments}\vskip 0.2in
\noindent #1}

\def\keywordname{{\bfseries \emph{Keywords}}}%
\def\keywords#1{\par\addvspace\medskipamount{\rightskip=0pt plus1cm
\def\and{\ifhmode\unskip\nobreak\fi\ $\cdot$
}\noindent\keywordname\enspace\ignorespaces#1\par}}

\begin{document}
\title{Random Forests for Change Point Detection}
\author{Malte Londschien${}^{1,2}$, Peter B\"uhlmann${}^{1}$, Solt Kov\'acs${}^{1}$\\
\vspace{0.1cm}\\
{\small${}^{1}$Seminar for Statistics, ETH Z\"urich, Switzerland}\\
{\small${}^{2}$AI Center, ETH Z\"urich, Switzerland}}

\date{May 2023}

\maketitle

\renewenvironment{abstract}{\centerline{\large\bf Abstract}\vspace{0.5ex}\begin{quote}}{\par\end{quote}\vskip 1ex} 

\begin{abstract}%
    We propose a novel multivariate nonparametric multiple change point detection method using classifiers.
    We construct a classifier log-likelihood ratio that uses class probability predictions to compare different change point configurations.    
    We propose a computationally feasible search method that is particularly well suited for random forests, denoted by \texttt{changeforest}.
    However, the method can be paired with any classifier that yields class probability predictions, which we illustrate by also using a $k$-nearest neighbor classifier.
    We prove that it consistently locates change points in single change point settings when paired with a consistent classifier.
    Our proposed method changeforest achieves improved empirical performance in an extensive simulation study compared to existing multivariate nonparametric change point detection methods.
    An efficient implementation of our method is made available for \texttt{R}, Python, and Rust users in the \texttt{changeforest} software package.
\end{abstract}


\section{Introduction}\label{seq:introduction}

Change point detection considers the localization of abrupt distributional changes in ordered observations, often time series.
We focus on offline problems, retrospectively detecting changes after all samples have been observed.
Inferring abrupt structural changes has a wide range of applications, including bioinformatics \citep{olshen2004circular, picard2005statistical}, neuroscience \citep{kaplan2001macrostructural}, biochemistry \citep{hotz2013idealizing}, climatology \citep{reeves2007review}, and finance \citep{kim2005structural}.
A rich literature has developed around statistical methods that recover change points in different scenarios, see \citet{truong2020selective} for a recent review.

Parametric change point detection methods typically assume that observations between change points stem from a finite-dimensional family of distributions.
Change points can then be estimated by maximizing a regularized log-likelihood over various segmentations.
The classical scenario of independent univariate Gaussian variables with constant variance and changes in the mean goes back to the 1950s \citep{page1954continuous, page1955test}.
It has recently been studied by, among others, \citet{frick2014multiscale}, \citet{fryzlewicz2014wild}, and references therein.
\citet{pein2017heterogeneous} consider a relaxed version that allows changes in the variance in addition to the mean shifts. 
Generalizations exist for multivariate scenarios.
Recently, even high-dimensional scenarios have been studied:
\cite{wang2018high} and \cite{enikeeva2019high} consider multivariate Gaussian observations with sparse mean shifts,
\citet{wang2021optimal} study the problem of changing covariance matrices of sub-Gaussian random vectors, \citet{roy2017change} study estimation in Markov random field models, and \citet{londschien2021change} consider time-varying Gaussian graphical models.

Nonparametric change point detection methods use measures that do not rely on parametric forms of the distribution or the nature of change.
Proposals for univariate nonparametric change point detection methods include \citet{pettitt1979non}, \citet{carlstein1988nonparametric}, \citet{dumbgen1991asymptotic}, and, more recently, \citet{zou2014nonparametric} and \citet{padilla2019optimal}.
Multivariate setups are challenging even in parametric scenarios.
The few multivariate nonparametric change point detection methods we are aware of are based on ranks \citep{lung2015homogeneity},
distances \citep{matteson2014nonparametric, chen2015graph, zhang2021graph},
kernel distances \citep{arlot2019kernel, garreau2018consistent, chang2019kernel},
kernel densities \citep{padilla2021optimal},
and relative density ratio estimation \citep{liu2013change}.

Many of these nonparametric proposals are related to hypothesis testing.
For a single change point, they maximize a test statistic that measures the dissimilarity of distributions.
In the modern era of statistics and machine learning, many nonparametric methods are available to learn complex conditional class probability distributions, such as random forests \citep{breiman2001random} and neural networks \citep{mcculloch1943logical}.
These have proven to often outperform simple rank and distance-based methods.
\citet{friedman2004multivariate} proposed to use binary classifiers for two-sample testing.
Recently, \citet{lopez2016revisiting} applied this framework in combination with neural networks and \citet{hediger2022use} with random forests.
Similar to such two-sample testing approaches, we use classifiers to construct a novel class of multivariate nonparametric  multiple change point detection methods.

\subsection{Our contribution}\label{sec:our_contribution}
We propose a novel classifier-based nonparametric change point detection framework.
Motivated by parametric change point detection, we construct a classifier log-likelihood ratio that uses class probability predictions to compare different change point configurations.
Theoretical results for the population case motivate the development of our algorithm.

We present a novel search method based on binary segmentation that requires a constant number of classifier fits to find a single change point.
We prove that this method consistently recovers a single change point when paired with a classifier that provides consistent class probability predictions.
For multiple change points, the number of classifier fits required scales linearly with the number of change points, making the algorithm highly efficient.
Our method is implemented for random forests and $k$-nearest neighbor classifiers in the \texttt{changeforest} software package, available for \texttt{R}, Python, and Rust users.
The algorithm achieves improved empirical performance compared to existing multivariate nonparametric change point detection methods.

\section{A nonparametric classifier log-likelihood ratio}
\label{sec:change_point_detection}
We are interested in detecting structural breaks in the probability distribution of a time series.
More formally, consider a sequence of independent random variables $(X_i)_{i = 1}^n \subset \BR^p$ with distributions $\tilde\BP_1, \dotsc, \tilde\BP_n$. %
Let
\begin{equation*}
\alpha^0 := \{0, n\} \cup \{i : \tilde\BP_i \neq \tilde\BP_{i + 1}\}
\end{equation*}
be the set of \emph{segment boundaries} and denote with $K^0 := |\alpha^0| - 1$ the total number of segments.
We label the elements of $\alpha^0$ by their natural order starting with zero.
Then consecutive elements in $\alpha^0$ define segments $\segkzero$ for $k = 1, \dotsc, K^0$ within which the map $i \mapsto \tilde\BP_i$ is constant and the $X_i \sim \BP_k \coloneqq \tilde\BP_{\alpha_k}$ are i.i.d.
We call $\alpha^0_1, \dotsc, \alpha^0_{K^0 \shortminus 1}$ the \emph{change points} of the sequence $X_1, \dotsc, X_n$.
We aim to estimate the change points (or equivalently $\alpha^0$) of the time series process $X_1, \dotsc, X_n$ upon observing a realization $x_1, \dotsc, x_n$.
We construct a classifier log-likelihood ratio and use it for change point detection.
We motivate this ratio later in Section~\ref{sec:non_parametric_methodology}, drawing parallels to parametric methods that We introduce in Section~\ref{sec:parametric_methodology}.

For any segmentation $\alpha$, let $\by_\alpha = (y_i)_{i=1}^n$ be the sequence such that $y_i = k$ whenever $i \in \segk$.
Consider a classification algorithm $\hat p$ that can produce class probability predictions and denote with $\hat p_\alpha$ the classifier trained on covariates $\bX = (x_i)_{i=1}^n$ and labels $\by_\alpha$.
Write $\hat p_\alpha(x_i)_k$ for the trained classifier's class $k$ probability prediction for observation $x_i$.
The \emph{classifier log-likelihood ratio} for change point detection is
\begin{equation}\label{eq:loglikelihood_non_parametric}
    G\left((x_i)_{i = 1}^n \mid \alpha, \hat p \right) 
    \coloneqq \sum_{k = 1}^K \sum_{i = \alpha_{k \shortminus 1} + 1}^{\alpha_k} \log( \frac{n}{\alpha_k - \alpha_{k \shortminus 1}}\hat p_\alpha(x_i)_k).
\end{equation}

Consider a setup with a single change point.
If a split candidate $s$ is close to that change point, we expect the trained classifier $\hat p_{\{0, s, n\}}$ to perform better than random at separating $x_1, \ldots x_s$ from $x_{s+1}, \ldots x_n$ and $G\left((x_i)_{i=1}^n \mid \{0, s, n\}, \hat p \right)$ will take a positive value.
For a segmentation $\alpha$ with multiple change points, the classifier log-likelihood ratio (\ref{eq:loglikelihood_non_parametric}) measures the separability of the segments $\seg{\alpha^0}{\alpha_1}, \ldots, \seg{\alpha_{K\shortminus 1}}{\alpha_K}$.
We expect this separability to be highest for $\alpha = \alpha^0$, which we formalize in Proposition \ref{lem:optimal_loglikelihood}.

\subsection{The parametric setting}\label{sec:parametric_methodology}
Since first proposals from \citet{page1954continuous}, parametric approaches are typically based on maximizing a parametric log-likelihood, see also \citet{frick2014multiscale} and references therein.
For this, one assumes that the
$X_i \sim \BP_k$ are i.i.d.\ within segments $i \in \segkzero$ and that the distributions
$\BP_k \in \{\BP_\theta \mid \theta \in \Theta\}$
for some finite-dimensional parameter space $\Theta$.
The distribution of
$(X_i)_{i=1}^n$
can then be parametrized with the tuple
$(\alpha^0, (\theta^0_k)_{k=1}^{K^0})$
such that for all
$k = 1, \dotsc, K^0$ and $i = \alpha^0_{k \shortminus 1} + 1, \dotsc, \alpha^0_k$
the $X_i$ are $\BP_{\theta^0_k}$-distributed.

Assuming the $\BP_\theta$ have densities $p_\theta$, we can express the log-likelihood
of observing some sequence $(x_i)_{i = 1}^n$ in a setup parametrized by
$(\alpha, (\theta_k)_{k=1}^{K})$ as
\begin{equation}\label{eq:change_point_loglikelihood}
    \ell\left((x_i)_{i = 1}^n \mid \alpha, (\theta_k)_{k=1}^{K}\right)
    =
    \sum_{k = 1}^{K}
    \sum_{i = \alpha_{k \shortminus 1} + 1} ^ {\alpha_k}
    \log(p_{\theta_k}(x_i)).
\end{equation}
An estimate of $\alpha^0$ can be obtained by maximizing (\ref{eq:change_point_loglikelihood}).
Since the log-likelihood is increasing in $K = |\alpha| - 1$, an additional penalty term $\gamma > 0$ has to be subtracted whenever the number of change points is unknown. Choose
\begin{equation}
    \label{eq:change_point_estimator}
    \hat \alpha_\gamma \in \underset{\alpha\in\CA}{\argmax}\ \underset{\theta_1, \ldots, \theta_K}{\max} \ \ell\left((x_i)_{i = 1}^n \mid \alpha, (\theta_k)_{k=1}^{K}\right) - |\alpha| \gamma \ ,
\end{equation}
where $\CA$ is a set of possible segmentations, typically restricted such that each segment has a minimum length of $\delta n$ for some $\delta>0$.

The most popular parametrization of in-segment distributions are Gaussian distributions with known variance and a changing mean, see for example \citet{yao1988estimating} and \citet{fryzlewicz2014wild}.
Next to changes in the mean, \citet{killick2012optimal} also study changes in the variance, where the mean of the Gaussian is unknown but constant over all segments.
\citet{kim1989likelihood} study changes in the parameters of a linear regression model.
\citet{shen2012change} and \citet{cleynen2014segmentation} consider changes in the rate of a Poisson and changes in the success rate in a negative binomial, with applications in detecting copy number alterations in DNA.
\citet{frick2014multiscale} consider the general case where in-segment distributions belong to an exponential family.

\subsection{The nonparametric setting}
\label{sec:non_parametric_methodology}
Often, a good parametric model is hard to justify and nonparametric methods can be applied instead.
Again, consider a classification algorithm $\hat p$ that yields class probability predictions after training and recall that $\by_\alpha = (y_i)_{i=1}^n$ with $y_i = k$ whenever $i \in \segk$ and that ${\bX = (x_i)_{i=1}^n}$.
For any $\alpha$, training $\hat p$ on $(\bX, \by_\alpha)$ yields the function $\hat p_\alpha \coloneqq \hat p(\bX, \by_\alpha)$ that assigns class $k$ probability estimates $\hat p_\alpha(x)_k$ corresponding to segments $\segk$ to observations $x$.
Write
$\BP_\uv := \frac{1}{v - u}\sum_{i=u+1}^v \tilde \BP_i$
for the mixture distribution of
$X_{u+1}, \ldots, X_v$
and assume that densities $p_\uv$ exist.
Assuming a uniform prior on class probabilities, the trained classifier $\hat p_\alpha$ estimates
\begin{equation*}
    \hat p_\alpha(x)_k \approx \BP(Y = k \mid X = x) = \frac{\dd \BP(X = x \mid Y = k)}{\dd \BP(X = x)}\BP(Y = k)
    = \frac{p_\seg{\alpha_{k \shortminus 1}}{\alpha_k}(x_i)}{p_\seg{0}{n}(x_i)}\frac{\alpha_k - \alpha_{k\shortminus 1}}{n}
\end{equation*}
for $k = 1, \ldots K$. This yields an approximation
\begin{equation*}
    \frac{p_\seg{\alpha_{k \shortminus 1}}{\alpha_k}(x) }{p_\seg{0}{n}(x)} \approx \frac{n}{\alpha_k - \alpha_{k \shortminus 1}} \hat p_\alpha(x)_k \,,
\end{equation*}
which we use in the definition of the classifier log-likelihood ratio (\ref{eq:loglikelihood_non_parametric}):
\begin{align*}
    G\left((x_i)_{i = 1}^n \mid \alpha, \hat p \right) 
    &= \sum_{k = 1}^K \sum_{i = \alpha_{k \shortminus 1} + 1}^{\alpha_k} \log( \frac{n}{\alpha_k - \alpha_{k \shortminus 1}}\hat p_\alpha(x_i)_k)
    \approx \sum_{k = 1}^K \sum_{i = \alpha_{k \shortminus 1} + 1}^{\alpha_k} \log(\frac{p_\segk(x_i)}{p_\seg{0}{n}(x_i)}) \, .
\end{align*}
This resembles the parametric log-likelihood (\ref{eq:change_point_loglikelihood}), where $p_\seg{\alpha^0_{k\shortminus 1}}{\alpha^0_k} = p_{\theta^0_k}$, up to the term
$-\sum_{i=1}^n \linebreak[1] \log(p_\seg{0}{n}(x_i))$, which is independent of $\alpha$.
In Equation~(\ref{eq:loglikelihood_non_parametric}), the classification algorithm $\hat p$ takes the role of the parametrization by $\theta \in \Theta$ in Equation~(\ref{eq:change_point_loglikelihood}).
Analogously to (\ref{eq:change_point_estimator}), a nonparametric estimator for $\alpha^0$ could be given by
\begin{equation}\label{eq:change_point_estimator_nonparametric}
    \hat \alpha_\gamma \in \argmax_{\alpha \in \CA} G\left( (x_i)_{i = 1}^n \mid \alpha, \hat p \right) - \gamma |\alpha|\ ,
\end{equation}
where $\CA$ is a suitable set of possible segmentations.

\begin{remark}
    To further motivate the classifier log-likelihood ratio (\ref{eq:loglikelihood_non_parametric}), take a generative classifier $\hat p$ (as in linear discriminant analysis).
    Choose $\hat \theta(\uv) \in \argmax_\theta \sum_{i=u+1}^v \log(p_\theta(x_i))$ and for any segmentation $\alpha$ let $\pi_k \coloneqq \frac{\alpha_{k\shortminus 1} - \alpha_k}{n}$, such that
    $\hat p_\alpha(x_i)_k =
    \frac{
        \pi_k \ p_{\hat\theta(\segk)}(x_i) 
    }{
        \sum_{j=1}^K \pi_j \ p_{\hat\theta(\segj)}(x_i)
    }$.
    Then the classifier log-likelihood ratio is equal to the parametric log-likelihood ratio between the segment distributions 
    $p_{\hat \theta(\segk)}$ and the mixture model $\sum_{j=1}^K  \pi_j \ p_{\hat \theta(\segj)}(\cdot)_j$:
    \begin{align*}
        G\left( (x_i)_{i = 1}^n \mid \alpha, \hat p \right)
        &= \sum_{k = 1}^K \sum_{i = \alpha_{k \shortminus 1} + 1}^{\alpha_k}
        \log( \frac{
            p_{\hat\theta(\segk)}(x_i) 
        }{
            \sum_{j=1}^K \pi_j p_{\hat\theta(\segj)}(x_i)
        })\\
        &=\ell\left((x_i)_{i = 1}^n \mid (\alpha, (\hat\theta(\segk))_{k=1}^{K})\right) - \sum_{n=1}^n \log( \sum_{j=1}^K \pi_j \ p_{\hat\theta(\segj)}(x_i) ) \ . \\
    \end{align*}
\end{remark}%
The following proposition suggests that for any estimator approximating the Bayes classifier, maximizing $G$ over segmentations $\alpha$ yields a reasonable estimator for $\alpha^0$.
\begin{restatable}{proposition_}{lemmaoptimalloglikelihood}
\label{lem:optimal_loglikelihood}
    Let \begin{equation*}
        p^\infty_\alpha(x) \coloneqq
        \BP\left(Y = k \mid X = x\right) =
        \left(
        \frac{\alpha_k - \alpha_{k \shortminus 1}}{n}
        \frac{p_\seg{\alpha_{k \shortminus 1}}{\alpha_k}}{p_\seg{0}{n}}
        \right)_{k = 1}^{|\alpha| - 1}
    \end{equation*}
    be the infinite sample Bayes classifier corresponding to the segmentation $\alpha$.
    If $\CA$ is a set of segmentations containing $\alpha^0$ and $\DKL{P}{Q}$ is the Kullback-Leibler divergence of $Q$ with respect to $P$, then
    \begin{align*}
    \max_{\alpha\in\CA} \,
    \BE\left[ G\left( (X_i)_{i = 1}^n \mid \alpha, p^\infty \right)\right] &= 
    \BE\left[ G\left( (X_i)_{i = 1}^n \mid \alpha^0, p^\infty \right)\right]\nonumber \\
    &= \sum_{k = 1}^K (\alpha^0_k - \alpha^0_{k \shortminus 1}) \DKL{\BP_k}{\BP_\seg{0}{n}}
    \end{align*}
    and any maximizer of $\BE\left[ G\left( (X_i)_{i = 1}^n \mid \alpha, p^\infty \right)\right]$ is equal to $\alpha^0$ or is a segmentation containing $\alpha^0$.
    In particular, the maximizer $\alpha$ with the smallest cardinality is equal to $\alpha^0$.
\end{restatable}
\section{The \texttt{changeforest} algorithm}\label{sec:implementation}

Finding a solution to the nonparametric estimator (\ref{eq:change_point_estimator_nonparametric}) by evaluating the classifier log-likelihood ratio~(\ref{eq:loglikelihood_non_parametric}) for all possible candidate segmentations $\alpha \in \CA$ is infeasible.
We present an efficient search procedure based on binary segmentation that finds an approximate solution to (\ref{eq:change_point_estimator_nonparametric}).
We furthermore motivate the use of random forests as classifiers and present a model selection procedure adapted to our search procedure.
These building blocks combine to form the \texttt{changeforest} algorithm, which we summarize in Section~\ref{sec:our_proposals}.
Finally, we present a consistency result for our method in the presence of a single change point.

\subsection{Binary segmentation}

Binary segmentation \citep{vostrikova1981detecting} is a popular greedy algorithm to obtain an approximate solution to the parametric maximum log-likelihood estimator (\ref{eq:change_point_estimator}).
For some $\hat \theta(\uv) \in \argmax_\theta \linebreak[1] \sum_{i=u+1}^v \log(p_\theta(x_i))$, binary segmentation recursively splits segments at the split~$s$ maximizing the gain
\begin{equation}\label{eq:parametric_gain}
    G_\uv(s) := \sum_{i = u+1}^s \log(\frac{p_{\hat\theta(\us)}(x_i)}{p_{\hat\theta(\uv)}(x_i)}) + \sum_{i=s+1}^v   \log(\frac{p_{\hat\theta(\sv)}(x_i)}{p_{\hat\theta(\uv)}(x_i)}) \, ,
\end{equation}
the increase in log-likelihood, until a stopping criterion is met.
For change in mean, where $\hat \theta(\uv) = \frac{1}{v - u}\sum_{i=u+1}^v x_i$ and $p_\theta(x) = \frac{1}{\sqrt{2\pi}}\exp(-\frac{1}{2}(x - \theta)^2)$,
the normalized gain $\frac{2}{v - u} G_\uv(s)$
is equal to
$\left(\sqrt{\frac{v - s}{s - u}} \sum_{i=u+1}^s x_i - \sqrt{\frac{s - u}{v - s}}\sum_{i=s+1}^v x_i \right)^2$, the square of the CUSUM statistic, first presented by \citet{page1954continuous}.
Binary segmentation typically requires $\CO(K n \log(n))$ evaluations of the gain, where $K$ is the number of change points, and is typically faster than search methods based on dynamic programming such as PELT \citep{killick2012optimal}.

We replace the parametric log-likelihood ratio $G_\uv(s)$ in binary segmentation with the nonparametric classifier log-likelihood ratio from (\ref{eq:loglikelihood_non_parametric})
\begin{align}\label{eq:nonparametric_gain}
    G\left( (x_i)_{i=u+1}^s \mid \{u, s, v\}, \hat p \right) &
    = \sum_{i=u+1}^s \log( \frac{v-u}{s-u} \hat p_{\{u, s, v\}}(x_i)_1 ) + 
    \sum_{i=s+1}^v \log( \frac{v-u}{v-s} \hat p_{\{u, s, v\}}(x_i)_2 ) \\
    &\approx \sum_{i=u+1}^s \log( \frac{p_{\us}(x_i)}{p_{\uv}(x_i)}) +
    \sum_{i=s+1}^v \log( \frac{p_{\sv}(x_i)}{p_{\uv}(x_i)}) \, . \nonumber
\end{align}

In many parametric settings, the expected gain curve $G_\uv$ will be piecewise convex between the underlying change points \citep{kovacs2020optimistic}.
Proposition \ref{prop:convexity} shows that the same holds for the nonparametric variant (\ref{eq:nonparametric_gain}) in the population case.
Figure~\ref{fig:gain_curves} shows the nonparametric gains for a random forest and a $k$-nearest neighbor classifier applied to simulated data with two change points.
\begin{restatable}{proposition_}{lemmaconvexity}\label{prop:convexity}
For the Bayes classifier $p^\infty$, the expected classifier log-likelihood ratio
\begin{equation*}
    s \mapsto \BE \left[ G \left( (X_i)_{i = u+1}^v \mid \{u, s, v\}, p^\infty \right) \right]
\end{equation*}
is piecewise convex between the underlying change points $\alpha^0$,
with strict convexity in the segment $\seg{\alpha^0_{k \shortminus 1}}{\alpha^0_k}$ if
$\BP_\seg{u}{\alpha^0_{k \shortminus 1}} \neq \BP_\seg{\alpha^0_{k\shortminus 1}}{\alpha^0_k}$ or
$\BP_\seg{\alpha^0_{k\shortminus 1}}{\alpha^0_k} \neq \BP_\seg{\alpha^0_k}{v}$.
\end{restatable}
\begin{figure}[h]
    \centering
    \includegraphics[width = \textwidth]{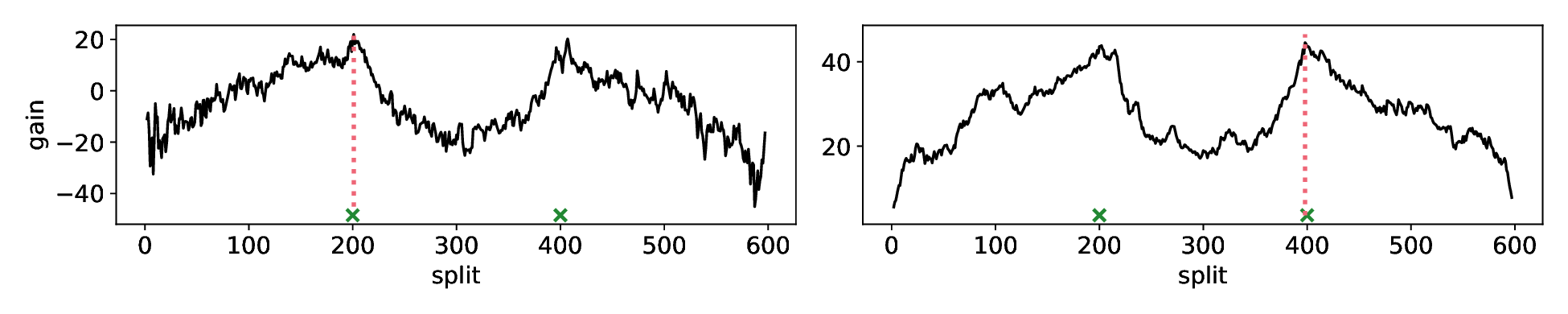}
    \caption{
        \label{fig:gain_curves}
        Classifier-based gain curves using a random forest (left) and a $k$-nearest neighbor classifier (right).
        The time series with 600 observations of dimension 5 and change points at $t=200, 400$ (green crosses) were simulated with changing covariances, as described in Section~\ref{sec:simulations_setups}.
        The maximizer is marked in red.
    }
\end{figure}

\subsection{The two-step search algorithm}\label{sec:two-step-search}
Binary segmentation relies on the full grid search, where $G_\uv(s)$ is evaluated for all ${s=u+1, \ldots, v}$, to find the maximizer of the gain $G_\uv(s)$.
In many traditional parametric settings, such as change in mean, log-likelihoods of neighboring segments can be recovered using cheap $\CO(1)$ updates.
This enables change point detection with binary segmentation in $\CO(K n \log(n))$ time, where $K$ is the number of change points.

For many classifiers, random forests in particular, such updates are unavailable and the classifiers need to be recomputed from scratch, making the computational cost of the grid search in binary segmentation prohibitive.
Similar computational issues arise in high-dimensional regression models.
Here, \citet{kaul2019efficient} propose a two-step approach.
Starting with an initial guess $s^{(0)}$, they fit a single high-dimensional linear regression for each of the segments $\seg{u}{s^{(0)}}$ and $\seg{s^{(0)}}{v}$ and
generate an improved estimate of the optimal split using the resulting residuals.
They repeat the procedure two times and show a result about consistency in the high-dimensional regression setting with a single change point if the change point is sufficiently close to the first guess $s^{(0)}$.

We apply a variant of the two-step search paired with the classifier log-likelihood ratio for multiple change point scenarios.
Instead of residuals, we recycle the class probability predictions and the resulting classifier log-likelihood ratios for individual observations from a single classifier fit $\hat p_{\{u, s^{(0)}, v\}}$, approximating
\begin{multline*}
    G\left((x_i)_{i=u+1}^v \mid \{ u, s, v\}, \hat p\right)
    \approx
    \sum_{i=u+1}^s \log(\frac{v - u}{s^{(0)} - u} \hat p_{\{u, s^{(0)}, v\}}(x_i)_1) +\\
    \sum_{i=s+1}^v \log(\frac{v - u}{v - s^{(0)}} \hat p_{\{u, s^{(0)}, v\}}(x_i)_2).
\end{multline*}
We call this the \emph{approximate gain} in the following.
Note that the classifier and normalization factors are fixed, but the summation varies with the split $s$.
Like \citet{kaul2019efficient}, we compute this twice, using the first maximizer of the approximate gain as a second guess $s^{(1)}$.
This allows us to find local maxima of the nonparametric gain (\ref{eq:nonparametric_gain}) in a constant number of classifier fits.

\citet{kaul2019efficient} only apply the two-step search in scenarios with a single change point.
In scenarios with multiple change points, as encountered in binary segmentation, the class probability predictions from the initial classifier fit might contain little information if $\BP_\seg{u}{s^{(0)}} \approx \BP_\seg{s^{(0)}}{v}$.
To avoid this, we start with multiple initial guesses and select the split point corresponding to the overall highest approximate gain as the second guess.
The resulting two-step algorithm, as implemented in \texttt{changeforest}, using three initial guesses at the segment's 25\%, 50\%, and 75\%-quantiles is presented in Algorithm~\ref{alg:two_step_search}.
Proposition \ref{lem:piece_wise_linear} suggests good estimation performance when coupling the two-step search with our classifier-based approach in single change point scenarios.
We present a consistency result of the two-step search when paired with a classifier that yields consistent class probability predictions for single change point scenarios in Section~\ref{sec:theoretical_guarantees}.

\begin{restatable}{proposition_}{lemmapiecewiselinear}
\label{lem:piece_wise_linear}
For the Bayes classifier $p^\infty$ and any initial guess $u < s^{(0)} < v$, the expected approximate gain
\begin{equation*}
    s \mapsto \BE\left[\sum_{i = u + 1}^{s} \log( \frac{v - u}{s^{(0)} - u} p_{\{u, s^{(0)}, v\}}^\infty(X_i)_1) + \sum_{i = s + 1}^v \log(\frac{v - u}{v - s^{(0)}} p_{\{u, s^{(0)}, v\}}^\infty(X_i)_2)\right]
\end{equation*}
is piecewise linear between the underlying change points $\alpha^0$.
If there is a single change point $a^0 \in \uv$, the expected approximate gain has a unique maximum at $s=a^0$.
\end{restatable}

Figure~\ref{fig:two_step_search} shows approximate gain curves and the underlying class probability predictions for a data set simulated with changing covariances, as described in Section~\ref{sec:simulations_setups}.
The approximate gain curves have a piecewise linear shape with kinks at the underlying change points, as predicted by Proposition~\ref{lem:piece_wise_linear}.
One can also observe how the distributions of the probability predictions change at the true change points.
For the second row, the choice of the initial guess $s^{(0)} = 300$ is precisely such that $\BP_\seg{0}{300} = \BP_\seg{300}{600}$, and the approximate gain curve is flat.
Figure~\ref{fig:binary_segmentation} shows approximate gain curves in different stages of binary segmentation for a different simulated time series.
Again, the approximate gain curves are roughly piecewise linear with kinks at the underlying change points.

\begin{figure}[h]
    \centering
    \includegraphics[width = \textwidth]{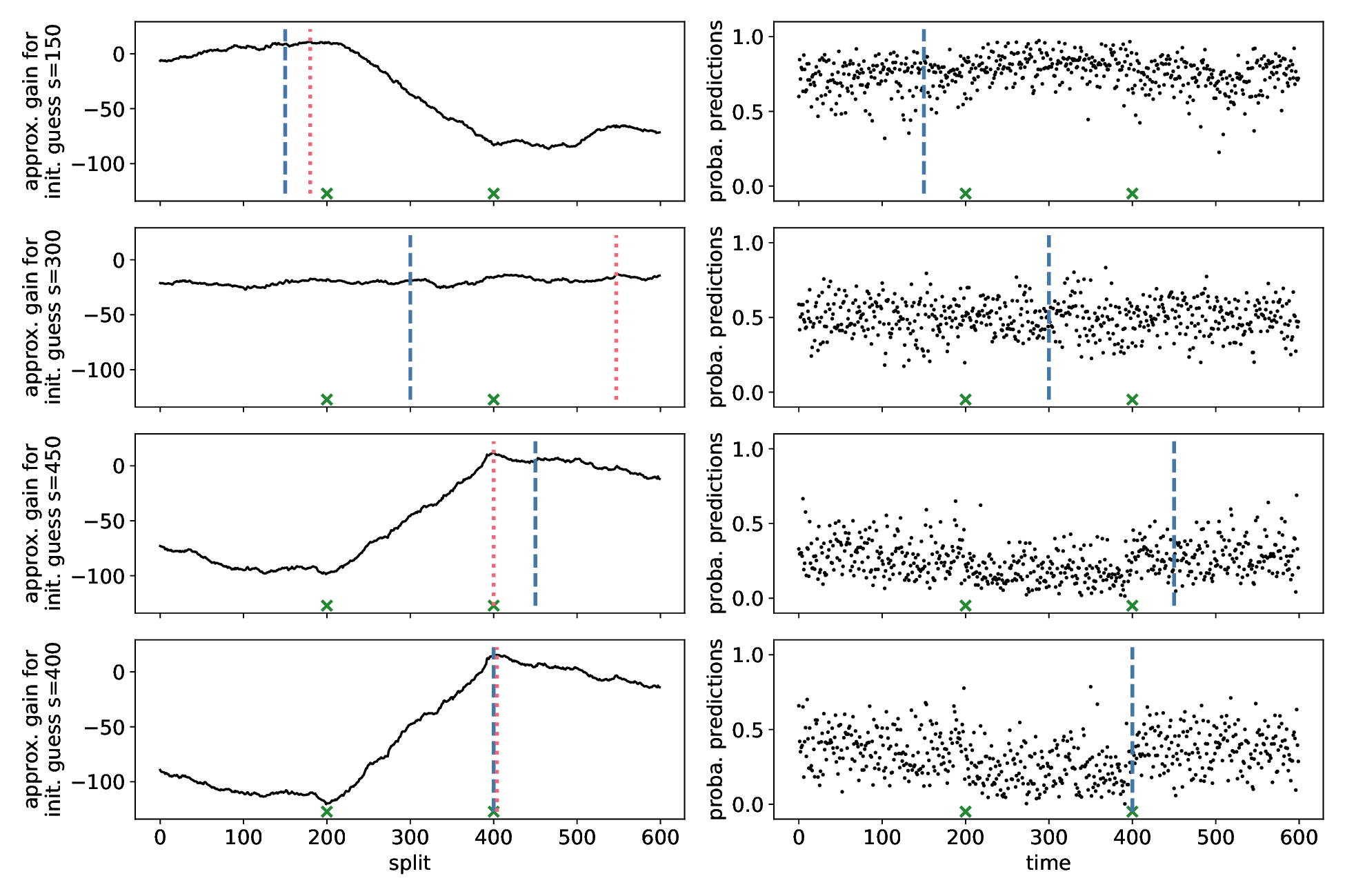}
    \vspace{-0.5cm}
    \caption{
        \label{fig:two_step_search}
        Approximate gain curves (left) from single random forest fits for initial guesses $s^{(0)} = 150$ (top), $300$ (second row), and $450$ (third row) and the underlying out-of-bag class probability predictions (right).
        The last row shows the approximate gain curve for the second guess $\hat s^{(1)} = 400$.
        The underlying data is the same as in Figure~\ref{fig:gain_curves}, with true change points at $t=200, 400$ (green crosses).
        The initial guess and the curve's maximizer are marked by blue and red vertical dashed lines, respectively.
    }
\end{figure}

\subsection{Choice of classifier}\label{sec:choice_of_classifier}
We recommend pairing our classifier-based method with random forests \citep{breiman2001random}. We motivate this choice.

A classifier should satisfy three properties when paired with our change point detection methodology to enjoy good estimation performance and computational feasibility:
(i) It needs to have a low dependence on hyperparameter tuning, (ii) it needs to be able to generate unbiased class probability predictions, and (iii) the training cost should be reasonable even for a large number of samples.
(i) The hyperparameters optimal for classification can vary between segments encountered in binary segmentation.
Retuning them, as is done by \citet{londschien2021change}, is infeasible.
Thus, the classifier needs to provide suitable class probability predictions even for suboptimal hyperparameters.
(ii) If the class probability predictions have a strong bias due to overfitting, the approximate gain curve will have a kink at the initial guess, see Figure~\ref{fig:two_step_search_biased}.
For a weak underlying signal, the initial guess might even be the maximizer of the curve, causing the two-step algorithm to fail in finding a true change point.
(iii) As for computational cost: Our optimization algorithm based on binary segmentation and the two-step search massively reduces the number of classifier fits required for change point estimation.
Still, hundreds of fits might be necessary.
For this not to become restrictive, the classifier's training cost should not explode when the number of samples gets large.

Random forests are a perfect fit: They are known to be an off-the-shelf classifier that fits most data without much hyperparameter tuning, they can generate unbiased probability estimates through out-of-bag prediction, and their training complexity is nearly linear in the sample size for a fixed maximal depth.

We highlight another attractive attribute of random forests: They are applicable to a wide range of types of data.
The need for nonparametric change point detection arises primarily when little is known about the data to be analyzed.
Else, a smartly chosen parametric method might be better suited.
In such situations, classifiers with broad applicability are optimal.
As a tree-based method, random forests are invariant to feature rescaling and robust to the inclusion of irrelevant features.
Due to these properties, random forests are popular as black box or off-the-shelf learning algorithms and are thus well suited for many data sets encountered in nonparametric change point detection.

While random forests are our classifier of choice, our method, including the two-step search, can be paired with any classification algorithm that can produce unbiased class probability predictions.
As an illustration, we pair our approach with the Euclidean norm $k$-nearest neighbor classifier.
It has a single hyperparameter $k$, which we set to the square root of the segment length.
Furthermore, the algorithm can be adapted to recover leave-one-out cross-validated probability estimates.
Our implementation has a computational complexity of $\CO(n^2)$ to compute all pairwise distances once, which we can recycle for all segments.
More efficient implementations exist \citep{arya1998optimal}.
However, the quadratic runtime and memory requirements are not an issue for the reasonably sized data sets we use for benchmarking.

Other choices of classifiers, such as neural networks, would also be possible when paired with a cross-validation procedure to obtain unbiased class probability predictions.
However, we do not consider them here due to their high computational cost.

\begin{figure}[h]
    \centering
    \includegraphics[width = \textwidth]{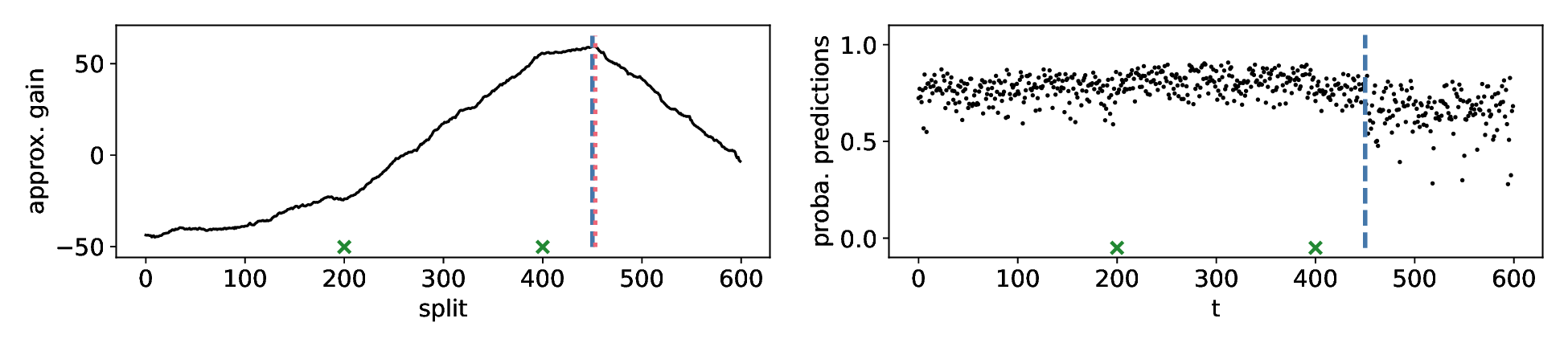}
    \vspace{-0.5cm}
    \caption{
        \label{fig:two_step_search_biased}
        In-sample (not out-of-bag) class probability predictions of a random forest (right) and the corresponding approximate gain curve (left).
        Predictions are from a single random forest fit for the initial guess $s^{(0)} = 450$ (blue line).
        The underlying data is the same as in Figures \ref{fig:gain_curves} and \ref{fig:two_step_search}, with true change points at $t=200, 400$ (green crosses).
        The maximizer is marked in red.
        The predictions show a clear bias, dominating the underlying signal (kink around the true change point at $t=400$), and the approximate gain curve has a peak at the initial guess.
        Compare this to the third row of Figure~\ref{fig:two_step_search}, where out-of-bag class probability predictions were used, resulting in a maximum around a true change point at $t=400$.
    }
\end{figure}

\subsection{Model selection}\label{sec:model_selection}
So far, we have discussed how to accurately and efficiently estimate the location of change points in a signal.
We now discuss how to decide whether to keep or discard a candidate split.

Separating true change points from false ones is crucial for good estimation performance and is a task as complex as finding candidate splits.
Two approaches for model selection are particularly popular in change point detection: thresholding and permutation tests.
For thresholding, a candidate split is kept if its inclusion results in an increase of the log-likelihood at least as large as the threshold $\gamma$.
Thresholding is often applied in parametric settings, where BIC-type criteria for the minimal gain to split can be derived with asymptotic theory, see for example \citet{yao1988estimating}.
Permutation tests are an alternative to thresholding when such asymptotic theory is unavailable.
Here, a statistic gathered for a candidate split is compared to values of the statistic gathered after the underlying data was permuted, see for example \citet{matteson2014nonparametric}.

No asymptotic theory is available to suggest a threshold for the (approximate) nonparametric classifier log-likelihood ratio presented in Section~\ref{sec:change_point_detection}.
Meanwhile, permutation tests require a high number of classifier fits in each splitting step, severely affecting the overall computational cost.
We apply a leaner pseudo-permutation test.
Instead of permuting the observations and refitting the classifier, we only permute the classifier’s predictions gathered in the first step of the two-step search, and thus the classifier log-likelihood ratios.
We then compute the approximate gains after shuffling and take their maxima.
In the absence of change points, there is no underlying signal for the classifier to learn other than the relative class sizes.
We thus expect the distribution of the classifier log-likelihood ratios to be approximately invariant under permutations.
This allows us to compare the maximal gain obtained in the first step of the two-step search to the maximal gains computed after permuting.
Algorithm~\ref{alg:model_selection} summarizes the procedure.

The significance threshold for the pseudo-permutation test is a tuning parameter rather than a valid significance level. We choose 0.02 by default.
This will typically lead to false positives due to multiple testing, as the permutation test is applied to each segment encountered in binary segmentation.
We use 199 permutations.
We analyze the empirical false positive rate of the pseudo-permutation test in Section~\ref{sec:false_positive_rates}.

\subsection{Our proposal}\label{sec:our_proposals} 
We present \texttt{changeforest} in Algorithm~\ref{alg:changeforest}.

We comment on two details: (i) We use $\log_\eta(x)$ instead of the natural logarithm and (ii) rescale probabilities with $1 \, / \, \pi_i^{\{u, s, v\}} = \frac{v - u - 1}{s - u - \indicator_{\{i \leq s\}}}$ instead of $\frac{v-u}{s-u}$.
(i) We introduce $\log_\eta(x)\coloneqq \log((1 - \eta)x + \eta)$ to use instead of $\log(x)$ in the definition of the classifier log-likelihood ratio.
Class probability predictions can take values 0 and 1, prohibiting us from using the logarithm directly to compute individual classifier log-likelihood ratios.
With $\eta = \exp(-6)$, we effectively cap individual classifier log-likelihood ratios from below by $-6$.
(ii) Concerning the rescaling: a random forest's out-of-bag probability predictions behave like those from leave-one-out cross-validation.
The prediction for observation $x_i$ was generated by a forest that was trained on the $v - u - 1$ observations $(x_j)_{j \in \{u+1, \ldots, v\} \backslash \{i\}}$.
If $i \leq s$, then $s - u - 1$ of these observations belong to the first class, and if $i >s$, $s - u$ do.
Thus, in the absence of any change points, we expect $\hat p_{\{u, s, v\}}(x_i) \approx \frac{s - u - \indicator_{\{i \leq s\}}}{v - u - 1} = \pi_i^{\{u, s, v\}}$ and thus scale predictions by the inverse of $\pi_i^{\{u, s, v\}}$ before taking the logarithm to recover the classifier log-likelihood ratio estimates.

Our algorithm is visualized in Figure~\ref{fig:binary_segmentation}, where we display the approximate gain curves of the two-step search for segments encountered in binary segmentation.

\begin{algorithm}[h]
    \caption{\texttt{changeforest}}
    \label{alg:changeforest}
    \begin{algorithmic}
        \State \textbf{Input: } \mbox{Observations $(x_i)_{i=1}^n$, a classifier $\hat p$, and a minimum relative segment length $\delta > 0$.}
        \State \textbf{Output: } Change point estimate $\hat \alpha \gets \textrm{BinarySegmentation}((x_i)_{i=1}^n, \hat p, \delta)$. \\
        \State \textbf{function }$\textrm{BinarySegmentation}((x_i)_{i=u+1}^v, \hat p, \delta)$
        \Indent
        \If{$v - u < 2 \delta n$}
            \State \textbf{return} $\varnothing$
        \EndIf \\

        \State $\hat s, ((\ell_{i, k, j})_{i=u+1, \ldots, v}^{k=1,2})_{j=1, 2, 3} \gets \textrm{TwoStepSearch}((x_i)_{i=u+1}^v, \hat p, \delta)$ \Comment{c.f. Algorithm~\ref{alg:two_step_search}}
        \State $q \gets \textrm{ModelSelection}(((\ell_{i, k, j})_{i=u+1, \ldots, v}^{k=1,2})_{j=1,2,3}, \delta)$ \Comment{c.f. Algorithm~\ref{alg:model_selection}}
        \\
        \If{$q \leq 0.02$}
            \State $\hat\alpha_\textrm{left} \gets \textrm{BinarySegmentation}((x_i)_{i=u+1}^{\hat s}, \hat p, \delta)$
            \State $\hat\alpha_\textrm{right} \gets \textrm{BinarySegmentation}((x_i)_{i=\hat s +1}^{v}, \hat p, \delta)$
            \State \textbf{return} $\hat\alpha_\textrm{left} \cup \{\hat s\} \cup \hat\alpha_\textrm{right}$
        \Else
            \State \textbf{return} $\varnothing$
        \EndIf\\
        \EndIndent
    \textbf{end function}
    \end{algorithmic}
\end{algorithm}
\begin{algorithm}[ht]
    \caption{TwoStepSearch}
    \label{alg:two_step_search}
    \begin{algorithmic}

        \State \textbf{Input: }\mbox{Observations $(x_i)_{i = u+1}^v$, a classifier $\hat p$, and a minimum relative segment length $\delta > 0$.}
        \State \textbf{Output: }A single change point estimate $\hat s$ and classifier log-likelihood ratios $((\ell_{i, k, j})_{i=u + 1, \ldots, v}^{k=1, 2})_{j=1, 2, 3}$ for model selection. \\
        \State \textbf{Step 1: } 
            \mbox{Define $(s^{(0)}_1, s^{(0)}_2, s^{(0)}_3) = (\lfloor \frac{3u + v}{4} \rfloor, \lfloor \frac{v + u}{2} \rfloor, \lfloor \frac{u + 3v}{4} \rfloor)$ and $\log_\eta(x) = \log((1 - \eta)x + \eta)$} for $\eta=\exp(-6)$ and set $\pi^{\{u, s, v\}}_i \gets \frac{s - u - \indicator_{\{i \leq s\}}}{v - u - 1}$. For $j = 1, 2, 3$, train binary classifiers $\hat p_{\{u, s^{(0)}_j, v\}}$ on $(\bX, \by_{\{u, s^{(0)}_j, v\}}) = ((x_i)_{i=u+1}^v, (1 + \indicator_{\{i > s^{(0)}_j\}})_{i=u+1}^v)$ and let
            \State \begin{equation*}
                \ell_{i, 1, j} =
                \log_\eta\left(
                    \hat p_{\{u, s^{(0)}_j, v\}}(x_i)_1 / \pi^{\{u, s^{(0)}_j, v\}}_i
                \right), \
                \ell_{i, 2, j} =
                \log_\eta\left(
                    \hat p_{\{u, s^{(0)}_j, v\}}(x_i)_2 / (1 - \pi^{\{u, s^{(0)}_j, v\}}_i)
                \right).
            \end{equation*}
            \State Select
            \begin{equation*}
                \hat s^{(1)} \in \argmax_{s = u + 1 + \lceil \delta n \rceil, \dotsc, v - \lceil \delta n \rceil} \max_{j = 1, 2, 3}
                \sum_{i = u+1}^s \ell_{i, 1, j} + \sum_{i = s + 1}^v \ell_{i, 2, j}.
            \end{equation*}
            \\
            \textbf{Step 2: }
            Train a binary classifier $\hat p_{\{u, \hat s^{(1)}, v\}}$ and let
            \State \begin{equation*}
                \ell_{i, 1} =
                \log_\eta\left(
                    \hat p_{\{u, s^{(1)}, v\}}(x_i)_1 / \pi^{\{u, s^{(1)}, v\}}_i
                \right), \
                \ell_{i, 2} =
                \log_\eta\left(
                    \hat p_{\{u, s^{(1)}, v\}}(x_i)_2 / (1 - \pi^{\{u, s^{(1)}, v\}}_i)
                \right).
            \end{equation*}
            \State Select
            \begin{equation*}
                \hat s = \hat s^{(2)} \in \argmax_{s = u + 1 + \lceil \delta n \rceil, \dotsc, v - \lceil \delta n \rceil} 
                \sum_{i = u+1}^s \ell_{i, 1} + \sum_{i = s + 1}^v \ell_{i, 2}.
            \end{equation*}

            \State \textbf{Return: } $\left(\hat s, ((\ell_{i, k, j})_{i=u + 1, \ldots, v}^{k=1, 2})_{j=1, 2, 3} \right)$
    \end{algorithmic}
\end{algorithm}

\begin{algorithm}[ht]
    \caption{ModelSelection}
    \label{alg:model_selection}
    \begin{algorithmic}
        \State \textbf{Input:} Classifier log-likelihood ratios $((\ell_{i, k, j})_{i=u + 1, \ldots, v}^{k=1, 2})_{j=1, 2, 3}$ from the two-step search and a minimum relative segment length $\delta > 0$.
        \State \textbf{Output:} An approximate $p$-value from the pseudo-permutation test. \\
        \vspace{-0.2cm}
        \State  \textbf{Step 1: } Recover the maximal gain encountered in the first step of the two-step search.
        \begin{equation*}
            G_0 = \max_{s=u+\lceil\delta n\rceil + 1, \ldots, v - \lceil\delta n\rceil} \ \max_{j=1, 2, 3}  \ \sum_{i=u+1}^s l_{i, 1, j} + \sum_{i=s+1}^v l_{i, 2, j}.
        \end{equation*}

        \State  \textbf{Step 2: } For random permutations $\sigma_l: \{u+1, \ldots, v\} \rightarrow \{u+1, \ldots, v\}$ for $l = 1 \ldots, 199$, compute maximal gains after permuting likelihoods
        \begin{equation*}
            G_l = \max_{j=1, 2, 3} \ \max_{s=u + \lceil \delta n \rceil +1, \ldots, v - \lceil \delta n \rceil} \ \sum_{i=u+1}^s l_{\sigma_l(i), 1, j} + \sum_{i=s+1}^v l_{\sigma_l(i), 2, j}.
        \end{equation*}\\
        \vspace{-0.2cm}
        \State \textbf{Return: } $\#\{l = 0, \ldots, 199 \mid G_l \geq G_0\} \, / \, 200$
    \end{algorithmic}
\end{algorithm}
\begin{figure}[ht]
    \centering
    \includegraphics[width = \textwidth]{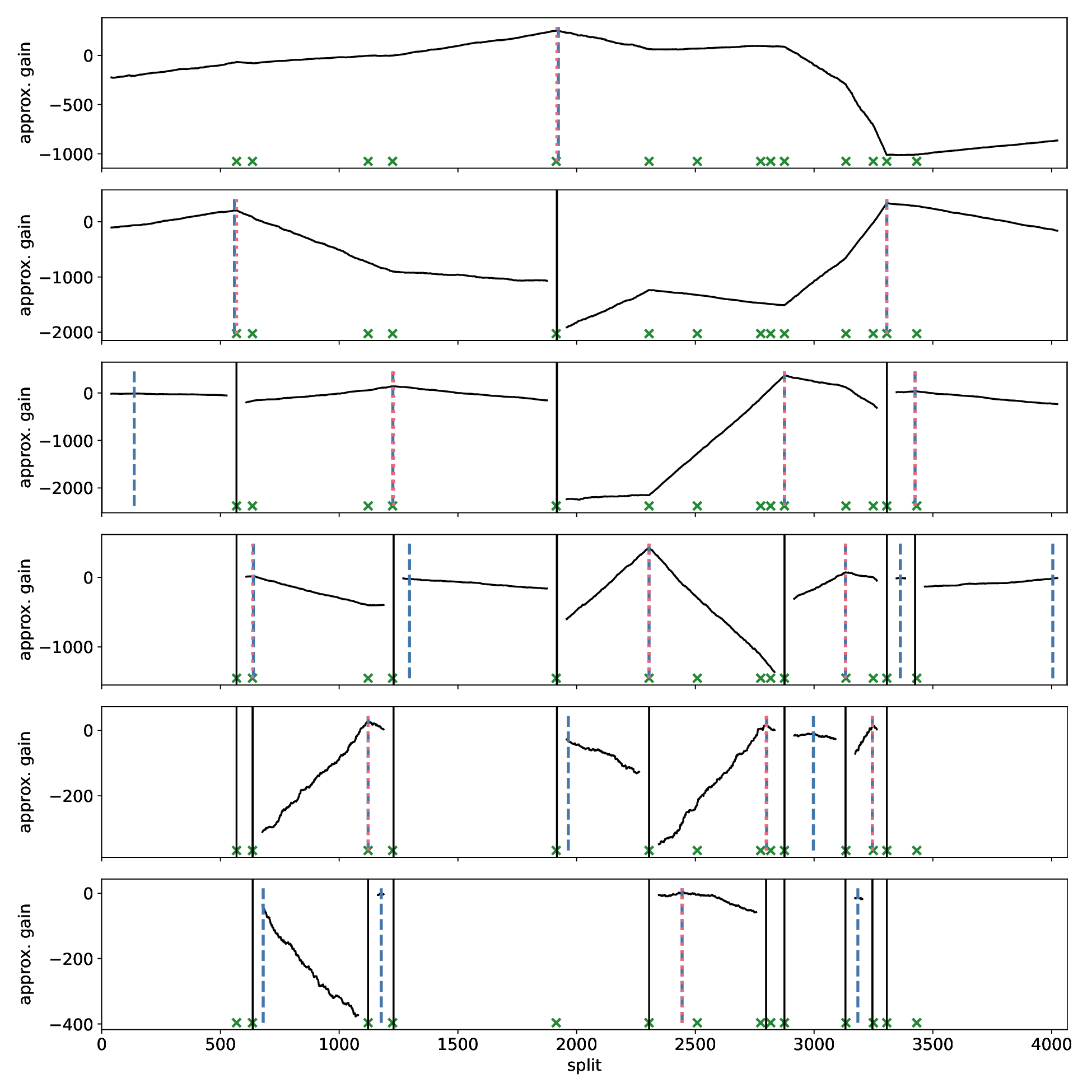}
    \caption{
        \label{fig:binary_segmentation}
        Approximate gain curves from the last step of the two-step search in the different steps of binary segmentation as encountered in \texttt{changeforest}.
        The underlying data set was simulated with the abalone setup, see Section~\ref{sec:simulations_setups} for details.
        The true change points at $t=568, 635, 1122, 1225, 1914, 2305, 2508, 2775, 2817, 2875, 3134, 3249, 3306, 3432$ are marked with green crosses.
        Change point estimates, corresponding to maxima of the approximate gain curves that are accepted in our model selection procedure, are marked in red.
        The second guess $\hat s^{(1)}$ for the two-step search is marked in blue.
        Current segment boundaries are marked with black vertical lines.
        Here, \texttt{changeforest} estimates change points at $t = 568, 636, 1122, 1229, 1917, 2305, 2444, 2798, 2875, 3132, 3245, 3306, 3425$, corresponding to an adjusted Rand index of 0.97.
    }
\end{figure}

\FloatBarrier

\subsection{Theoretical guarantees}
\label{sec:theoretical_guarantees}

We provide the following finite sample guarantee for the two-step search Algorithm \ref{alg:two_step_search} in settings with a single change point.

\begin{restatable}{theorem_}{finitesample}
    \label{prop:finitesample}
        Let $0 < \delta, \eta < 1$ and $\delta \leq \tau^0, \rho \leq 1 - \delta$. Set $s = \rho n$.
        Assume a single change point scenario with $n$ observations and a change point at $a^0 = \tau^0 n$.
        Let $\hat p$ satisfy $\max_{i=1, \ldots, n} | \hat p_s(X_i)_1 - p_s^\infty(X_i)_1| \overset{n\to\infty}{\rightarrow} 0$ in probability.
        Choose
        $$
            \hat a \in \argmax_{t=1, \ldots, n} \sum_{i = 1}^t \log_\eta\left(\frac{n}{s} \hat p_s(X_i)_1\right) + \sum_{i = t+1}^n \log_\eta\left(\frac{n}{n-s} \hat p_s(X_i)_2\right)
        $$
        as in Algorithm \ref{alg:two_step_search}.
        Then
        $$
        \BP \left[ \left|\frac{\hat a}{n} - \tau^0 \right| > \frac{C \log(n)}{\delta^2 \dTV{\BP_1}{\BP_2}^4 n} \right] \overset{n \to \infty}{\longrightarrow} 0,
        $$
        where $\dTV{\BP_1}{\BP_2}$ is the total variation distance between the two distributions $\BP_1$ and $\BP_2$ before and after the change point and
        $C = \frac{1}{(1 - \eta)^4} \log( \frac{1}{(1 - \rho)\rho} \frac{1}{\eta^2})^2$.
\end{restatable}

\noindent In particular, for our choices of $\rho \in \{\frac{1}{4}, \frac{1}{2}, \frac{3}{4} \}$ and $\eta=\exp(-6)$, the estimator $\frac{\hat\alpha}{n}$ is consistent at the rate $\frac{\log(n)}{n}$ for $\tau^0$.
Interestingly, this convergence rate is independent of the consistency rate of the classifier.
Generalizing Theorem \ref{prop:finitesample} to the case of multiple change points would require further assumptions, preventing the possible "canceling-out" effects as seen in Figure \ref{fig:two_step_search} (second row).
Theorem \ref{prop:finitesample} is based on a uniform consistency assumption for the classifier.
Such uniformly consistent classifiers exist \citep{bierens1983uniform}, including $k$-nearest neighbor classifiers \citep{kudraszow2013uniform}.
However, we are not aware of any such results for random forests.

\subsection{The \texttt{changeforest} software package}
A substantial part of this work is the \texttt{changeforest} software package, implementing Algorithm~\ref{alg:changeforest} for random forests and $k$-nearest neighbor classifiers.
It also implements the parametric change in mean setup, paired with binary segmentation, which we use as a baseline in the simulations.

The \texttt{changeforest} package is available for Python users on PyPI and conda-forge, \texttt{R} users on conda-forge, and Rust users on crates.io.
Its backend is implemented in the system programming language Rust \citep{matsakis2014rust}.
For inquiries, installation instructions, a tutorial, and more information, please visit \url{github.com/mlondschien/changeforest}.

\section{Empirical results}\label{sec:simulations}
We present results from a simulation study comparing the empirical performance of our classifier-based method \texttt{changeforest} to available multivariate nonparametric competitors.
The source code for the simulations, together with guidance on reproducing tables and figures, is available at \url{github.com/mlondschien/changeforest-simulations}.
The repository can easily be expanded to benchmark new methods provided by users.

\subsection{Competing methods}\label{sec:simulations_competing_methods}

We are aware of the following nonparametric methods for multivariate multiple change point detection with existing reasonably efficient implementations:
\citeauthor{matteson2014nonparametric}' \citeyearpar{matteson2014nonparametric} \emph{e-divisive} (ECP),
\citeauthor{lung2015homogeneity}'s \citeyearpar{lung2015homogeneity} \emph{MultiRank}, and kernel change point methods such as
\citeauthor{arlot2019kernel}'s \citeyearpar{arlot2019kernel} \emph{KCP} and \citeauthor{padilla2021optimal}'s \citeyearpar{padilla2021optimal} \emph{MNWBS}.

ECP searches for a single change point minimizing energy distances within segments.
The method generalizes to multiple change point scenarios through binary segmentation, and the significance of change points is assessed with a permutation test.
An implementation of ECP is available through the \texttt{R}-package \texttt{ecp} \citep{james2013ecp}, which we use in the simulations.

MultiRank uses a rank-based multivariate homogeneity statistic combined with dynamic programming.
The significance of change points is assessed based on asymptotic theory.
To run simulations with MultiRank, we used code made available to us by the authors.
This is included in the simulations repository.

Kernel change point methods minimize a within-segment average kernelized dissimilarity measure. 
\citet{arlot2019kernel} propose KCP, an algorithm based on dynamic programming.
An efficient implementation can be found in the Python package \texttt{ruptures} \citep{truong2020selective}.
We used the slope heuristic proposed by \citet{arlot2019kernel} for model selection.

MNWBS uses a kernel-density-based nonparametric extension of the CUSUM statistic for single change point localization.
The method extends to multiple change point scenarios through wild binary segmentation \citep{fryzlewicz2014wild}.
Model selection is done by thresholding of the test statistic, where the threshold is chosen based on the Kolmogorov–Smirnov test statistic.

For all packages, we used default parameters, if available. 
For ECP, this results in $\alpha=1$ and using 199 permutations at a significance level of 0.05.
For KCP, we use a Gaussian kernel with a bandwidth of 0.1 after normalization by the median absolute deviation of absolute consecutive differences, see Section~\ref{sec:simulations_setups}.
This choice of bandwidth was optimal for the simulated scenarios, see also Section~\ref{sec:tuning} and Table~\ref{tab:tuning_parameters_kcp}.
For MNWBS, we used the bandwidth $5 \cdot (n \log(n) / \delta)^{1 / p}$ as proposed by \citet{padilla2021optimal}.
We used 50 random intervals to reduce the computational cost.
We supplied information about the minimum relative segment length $\delta$, equal to 0.01 in the main simulations, to each method, either by informing about the minimum number of observations per segment $\lceil \delta n \rceil$ or by informing about the maximum number of change points $\lfloor 1 / \delta \rfloor$.

Other multivariate nonparametric multiple change point detection methods that we did not include in our simulations include \citeauthor{liu2013change}'s \citeyearpar{liu2013change} \emph{RuLSIF}, \citeauthor{cabrieto2017detecting}'s \citeyearpar{cabrieto2017detecting} \emph{Decon} with a kernel-based running statistic, and \citeauthor{zhang2021graph}'s \citeyearpar{zhang2021graph} \emph{gMulti}.
For gMulti, no implementation has been made available by the time of writing.
The existing implementation for Decon in the R package \textsf{kcpRS} does not implement the kernel-based nonparametric running statistic.
The implementation for RuLSIF is only available in Matlab and does not include an automatic model selection procedure.

\subsection{Simulation setups}\label{sec:simulations_setups}
We evaluate all methods on three parametric and six nonparametric scenarios.
We display their key characteristics below the header in Table \ref{tab:main_results}.

As for parametric setups, we use the \emph{change in mean} (CIM) and \emph{change in covariance} (CIC) setups from \citet[Section 4.3]{matteson2014nonparametric} and the \emph{Dirichlet} setup from \citet[Section 6.1]{arlot2019kernel}.
For both works, the authors present empirical results of their methods on data sets for three sets of parameters.
We use the set of parameters corresponding to medium difficulty.

For the change in mean and change in covariance setups, we generate time series of dimension $d=5$ with $n=600$ observations and change points at $t = 200, 400$.
Observations in the first and last segment are independently drawn from a standard multivariate Gaussian distribution.
In the change in mean setup, entries in the second segment are i.i.d.\ Gaussian with a mean shift of $\mu = 2$ in each coordinate.
In the change in covariance setup, entries in the second segment are i.i.d.\ Gaussian with mean zero and unit variance, but with a covariance of $\rho = 0.7$ between coordinates.
Lastly, the Dirichlet data set consists of $n = 1000$ observations and has dimension $d = 20$.
Ten change points are located at $t = 100, 130, 220, 320, 370, 520, 620, 740, 790, 870$.
In each segment, the observations are distributed according to the Dirichlet distribution, with the $d$ parameters drawn i.i.d.\ uniformly from $[0, 0.2]$.

We also generate time series with nonparametric in-segment distributions.
For this, we use popular multiclass classification data sets and treat observations corresponding to a single class as i.i.d.\ draws from a class-specific distribution.
We discard classes with fewer than $\delta n = \frac{n}{100}$ observations and shuffle and randomly concatenate the remaining classes.
Categorical variables are dummy encoded.
Finally, each covariate gets normalized by the median absolute deviation of absolute consecutive differences, as is also used by \citet{fryzlewicz2014wild} and further discussed by \citet{kovacs2020discussion}.
This normalization does not affect our random forest-based method \texttt{changeforest} and rank-based MultiRank but increases performance for distance-based competitors ECP, KCP, and MNWBS.
We did not normalize covariates in the parametric simulations setup, as this would decrease performance for ECP, KCP, and MNWBS.

We use the following data sets:
The \emph{iris} flower data set \citep{anderson1936species} contains 150 samples of four measurements from three species of the iris flower.
The \emph{glass}  identification data set \citep{evett1989rule}, initially motivated by criminological investigation, contains measurements of chemical compounds for different types of glass.
The \emph{wine}  data set \citep{cortez2009modeling} is the result of a physicochemical analysis of both red and white vinho verde from north-western Portugal.
It includes quality scores from blind testing, which we use as class labels. We encode the wine color as binary.
The Wisconsin \emph{breast cancer} data set \citep{street1993nuclear}  contains nominal characteristics of cell nuclei computed from a fine needle aspirate of breast tissue labeled as malignant or benign.
The resulting simulated time series has a single change point.
The \emph{abalone} data set \citep{waugh1995extending} contains easily obtainable physical measurements of snails and their age, determined by cutting through the shell cone and counting the number of rings.
Two categorical variables are dummy encoded.
We use the number of rings as class labels.
The \emph{dry beans} data set \citep{koklu2020multiclass} contains shape measurements derived from images of seven different types of dry beans.

The change in mean and change in covariance and possibly also the abalone and wine data sets feature similar or equal distributions in non-neighboring segments.
We expect these scenarios to be challenging for our methodology, see Section~\ref{sec:two-step-search} and Figure~\ref{fig:two_step_search}.

\subsection{Performance measures}
We mainly use the adjusted Rand index \citep{hubert1985comparing}, a standard measure to compare clusterings, to evaluate the goodness of fit of estimated segmentations.
Given two partitionings of $n$ observations, the Rand index \citep{rand1971objective} is the number of agreements, pairs of observations either in the same or in different subsets for both partitionings, divided by the total number of pairs $\binom{n}{2}$.
The adjusted Rand index is the normalized difference between the Rand index and its expectation when choosing partitions randomly.
Consequently, the adjusted Rand index can take negative values but no values greater than one.
Its expected value is zero for random class assignments, but the expected value for random segmentations, namely class assignments consistent with the time structure, is significantly higher.
We present a set of segmentations together with their adjusted Rand indices in Table \ref{tab:adj_rand_examples} in the Appendix.

We believe that an adjusted Rand index above 0.95 corresponds to a segmentation that is close to perfect.
An adjusted Rand index of 0.9 corresponds to a good segmentation, possibly oversegmenting with a false positive close to a true change point or missing one of two close change points.
On the other hand, we expect a segmentation with an adjusted Rand index significantly below 0.8 not to provide much scientific value. 

We additionally use the Hausdorff distance to evaluate change point estimates.
For two segmentations $\alpha, \alpha'$, set $d(\alpha, \alpha') \coloneqq \frac{1}{n} \max_{a\in\alpha} \min_{a'\in\alpha'} | a - a' |$.
Then $d(\alpha^0, \hat \alpha)$ is the largest relative distance of a true change point in $\alpha^0$ to its closest counterpart in $\hat \alpha$, and $d(\hat \alpha, \alpha^0)$ is the largest relative distance of a found change point in $\hat \alpha$ to the closest entry in $\alpha^0$.
Thus $d(\alpha^0, \hat \alpha)$ especially penalizes undersegmentation while $d(\hat \alpha, \alpha)$ penalizes oversegmentation.
Define the Hausdorff distance as $d_H(\alpha, \alpha') \coloneqq \max(d(\alpha, \alpha'), d(\alpha', \alpha))$.
Table \ref{tab:adj_rand_examples} also includes the Hausdorff distance for the selected segmentations.

\subsection{Main simulation results}
We present the results of our main simulation study in Table \ref{tab:main_results}, where we display average adjusted Rand indices over 500 simulations.
\begin{table}[h]
    \caption{
        \label{tab:main_results}
        Average adjusted Rand indices over 500 simulations, with standard deviations in parentheses.
        Optimal scores are marked in bold.
        The MultiRank method was unable to produce results for the dry beans simulation setup.
        For the abalone, wine, and dry beans simulation setups, the computational cost of MNWBS was prohibitive.
        Average and worst denote the average and worst performances across all setups.
    }
    \makebox[\textwidth][c]{
    \begin{tabular}{lrrrrr}
                       &            CIM         &            CIC            &                 Dirichlet &                      iris &               glass \\
        \midrule
        $n, d, K$ & $600, 5, 3$ & $600, 5, 3$ & $1000, 20, 11$ & $150, 4, 3$ & $214, 8, 6$ \\
        \midrule
        change in mean &            1.00 (0.00) &               0.00 (0.00) &               0.00 (0.00) &               0.93 (0.08) &               0.49 (0.17) \\
        changekNN      &            0.99 (0.02) &               0.01 (0.08) &               0.69 (0.25) &               0.99 (0.03) &               0.88 (0.09) \\
        ECP            &            0.99 (0.03) &               0.40 (0.32) &               0.84 (0.08) &               0.99 (0.03) &               0.61 (0.24) \\
        KCP            &            1.00 (0.01) &      \textbf{0.95} (0.06) &               0.87 (0.08) &      \textbf{0.99} (0.02) &               0.74 (0.23) \\
        MultiRank      &   \textbf{1.00} (0.00) &               0.00 (0.04) &               0.97 (0.02) &               0.57 (0.26) &               0.60 (0.09) \\
        MNWBS          &      0.99 (0.06) &               0.69 (0.18) &               0.67 (0.10) &               0.99 (0.04) &               0.71 (0.18) \\
        \texttt{changeforest}   &   0.99 (0.03) &               0.93 (0.12) &      \textbf{0.99} (0.02) &               0.98 (0.04) &               \textbf{0.92} (0.07) \\
    \end{tabular}
    }
    \newline
    \vspace*{0.2cm}
    \newline
    \makebox[\textwidth][c]{
    \begin{tabular}{lrrrrrr}
        &  breast cancer &        abalone &           wine &      dry beans &        average & worst \\
        \midrule
        $n, d, K$ & $699, 9, 2$ & $4066, 9, 15$ & $6462, 12, 5$ & $13611, 16, 7$ & & \\
        \midrule
        change in mean &            0.79 (0.03) &           0.88 (0.06) &               0.98 (0.02) &               0.99 (0.01) &               0.68 (0.07) &          0.00 \\
        changekNN      &            0.91 (0.17) &           0.88 (0.10) &               0.97 (0.04) &               1.00 (0.01) &               0.81 (0.11) &          0.01 \\
        ECP            &            1.00 (0.02) &           0.87 (0.06) &               0.98 (0.03) &               \textbf{1.00} (0.00) &      0.86 (0.14) &          0.43 \\
        KCP            &            \textbf{1.00} (0.03) &  0.92 (0.04) &               0.99 (0.01) &               \textbf{1.00} (0.00) &      0.94 (0.09) &          0.74 \\
        MultiRank      &            0.92 (0.12) &           0.81 (0.12) &               0.96 (0.03) &                            &              0.73 (0.11) &          0.00 \\
        MNWBS          &            0.99 (0.05) &                       &                           &                           &               0.84 (0.12) &          0.67 \\
        \texttt{changeforest}   &            0.98 (0.07) &           \textbf{0.93} (0.05) &      \textbf{0.99} (0.02) &      1.00 (0.01) &               \textbf{0.97} (0.06) &          \textbf{0.92} \\
    \end{tabular}
    }
    \vspace*{0.2cm}
    
\end{table}

Our method, \texttt{changeforest}, scores above 0.9 on average for each simulation scenario studied.
All other methods have an average score lower than 0.75 for at least one scenario.
\texttt{changeforest} is the best-performing method for four setups: Dirichlet, glass, abalone, and wine.
Furthermore, in four further setups, change in mean, iris, breast cancer, and dry beans, our method \texttt{changeforest} scores above 0.975 on average, corresponding to an almost perfect segmentation.
The only setup for which \texttt{changeforest} does not perform best or above 0.975 is change in covariance.
Here, \texttt{changeforest} scores second-best behind KCP and the only other method to score above 0.5 is MNWBS.
We compute the average score for each method to combine the results for the nine simulation scenarios.
\texttt{changeforest} scores highest.
Overall, \texttt{changeforest} is widely applicable, with robust performance on all parametric and nonparametric setups.

We include Table \ref{tab:main_results_hausdorff_median} with median Hausdorff distances in the Appendix.
The results are similar, with \texttt{changeforest} being among the methods with the best score, except for change in covariance, where it achieves a median Hausdorff distance of 0.013, behind KCP with a median Hausdorff distance of 0.007.

Table \ref{tab:main_results_n_cpts} in the Appendix shows the average number of estimated change points for each method and simulation setup.
\texttt{changeforest} slightly oversegments for most setups, but not for the glass and wine setups.

Figure~\ref{fig:histograms} in the Appendix shows histograms of the nonparametric methods' change point estimates for the Dirichlet setup, where the underlying change points are constant across simulations.
Corresponding to high adjusted Rand indices of 0.99 and 0.97, change point estimates by \texttt{changeforest} and MultiRank are sharply concentrated around the true change point locations.
Meanwhile, estimates by changekNN, ECP, KCP, and MNWBS are visually more spread out, corresponding to worse adjusted Rand indices smaller than 0.9.

We present average computation times for each method in Table \ref{tab:computation_times}.
Simulations were run on eight Intel Xeon 2.3 GHz cores with 4 GB of RAM available per core (32 GB in total).
The simple change in mean method is blazingly fast.
On the largest simulation setup dry beans with $n = 13611$ and $d = 16$, the average computation time was around 0.002 seconds.
\texttt{changeforest} is the fastest nonparametric method on the dry beans simulation setup, requiring around 2.6 seconds on average.
On the other hand, MNWBS and ECP are prohibitively slow.
MNWBS takes around four hours on average to estimate change points on the Dirichlet simulation setup.
It did not finish in a reasonable time on the larger dry beans, wine, and abalone simulation setups.
ECP is faster but still requires more than one hour to estimate change points on the dry beans simulation setup on average.
The code for MultiRank is experimental and broke for the dry beans simulation setup due to the multiple dummy encoded columns.
We further analyze the computational efficiency of available implementations in Section~\ref{sec:evolution_performance}.
\begin{table}[h]
    \caption{
        \label{tab:computation_times}
        Average computation times in seconds on 8 Intel Xeon 2.3 GHz cores. 
        }
    \vspace*{0.2cm}
    \begin{small}
    \makebox[\textwidth][c]{
        \begin{tabular}{lrrrrrrrrr}
            &   CIM &   CIC & Dirichlet &  iris & glass & breast & abalone &  wine & dry \\
                    &       &       &           &       &       & cancer &         &       & beans \\
            \midrule
            $n, d$  &  600, 5 & 600, 5 & 1000, 20 & 150, 4 & 214, 8 & 699, 9 & 4066, 9 & 6462, 12 & 13611, 16  \\
            \midrule
            change in mean      &  <0.01 &  <0.01 &      <0.01 &  <0.01 &  <0.01 &          <0.01 &    <0.01 &  <0.01 &      <0.01 \\
            changekNN           &  0.04 &  0.04 &      0.15 &  <0.01 &  0.01 &          0.05 &     1.8 &   4.7 &        24 \\
            ECP                 &   4.2 &   4.5 &        18 &  0.37 &  0.72 &           4.3 &     386 &   741 &      5356 \\
            KCP           &  0.08 &  0.09 &      0.18 &  0.02 &  0.04 &          0.10 &     2.5 &   6.6 &        31 \\
            MultiRank           &  0.77 &  0.91 &       2.1 &  0.10 &  0.17 &           1.0 &      34 &    90 &           \\
            MNWBS        &         2132 & 2143 & 14261 &        115 &   207 &        3211  &          &        & \\
            \texttt{changeforest} &  0.08 &  0.10 &      0.43 &  0.03 &  0.06 &          0.07 &    0.86 &  0.89 &       2.6 \\
        \end{tabular}
    }
    \end{small}
\end{table}

\subsection{Performance for varying numbers of segments and samples}\label{sec:evolution_performance}
The performance of change point estimation methods depends on an interplay between the number of observations, the size of distributional shifts between segments, and the number of change points.
In Table \ref{tab:main_results} we see that for some data sets, such as CIM, iris, and dry beans, all methods achieve average adjusted Rand indices above 0.975.
Here, we believe that the signal-to-noise ratio is so high that change point estimation is easy, and the difference between the average adjusted Rand index and the perfect score of 1.0 is primarily a result of random false positives.
On the other hand, the change in covariance and glass setups appear truly difficult, clearly separating methods that perform well (above 0.9) and those performing not much better than random guessing (below 0.8).

We introduce two simulation settings that let us generate time series of any length $n$ with any number of change points $K$, allowing us to control the underlying signal-to-noise ratio.
We draw $\tilde N_k \sim \Exp(1)$ i.i.d.\ for $k=1, \ldots, K$ and define $N_i = \frac{1}{10 K} + \frac{0.9}{\sum_{j=1}^K \tilde N_j} \tilde N_k$.
We use $\textrm{round}(n N_k)$ as segment lengths, where we round such that $\sum_{k=1}^K \textrm{round}(n N_k) = n$.
This results in a minimum relative segment length of $\delta = \frac{1}{10K}$.

We present two simulation setups: One based on the Dirichlet distribution and one based on the dry beans data set.
We construct the former as in Section~\ref{sec:simulations_setups}.
For each segment, we draw i.i.d.\ observations from a Dirichlet distribution of dimension $d=20$,  after sampling its parameters uniformly from $[0, 0.2]$.
To generate time series of arbitrary length based on the dry beans data set, we first normalize each covariate to have an average within-class variance of one.
Then, for each segment, we draw a class label different to that of the previous segment and draw $\textrm{round}(n N_k)$ observations with replacement from the corresponding class.
We add i.i.d.\ standard Gaussian noise to each observation.
As in Section~\ref{sec:simulations_setups}, we finally normalize each covariate by the median absolute deviation of absolute consecutive differences.
We simulate 500 data sets with $K = 20, 80$ and vary $n = 250, 354, 500, 707, 1000, \dotsc, 64000$ for each setup.
The mean adjusted Rand indices of change point estimates from our and competing methods are displayed in Figure~\ref{fig:evolution_performance}.
Each method was supplied with the minimum relative segment length $\delta = \frac{1}{10K}$.
\begin{figure}[h]
    \centering
    \includegraphics[width = \textwidth]{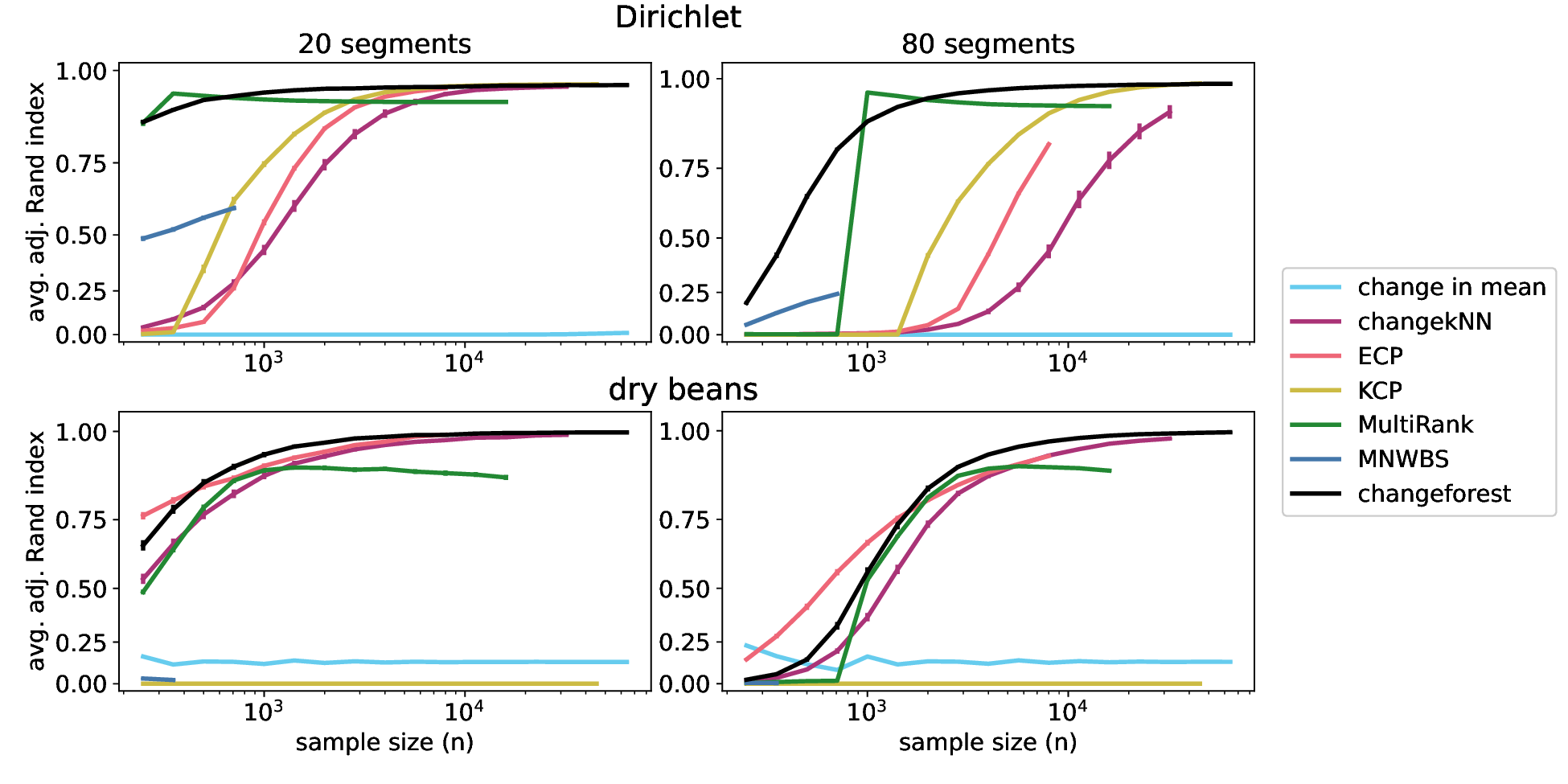}
    \vspace{-0.5cm}
    \caption{
        \label{fig:evolution_performance}
        Evolution of the average adjusted Rand index for methods by the number of observations for 20 and 80 segments.
        The $y$-axis was scaled to better pronounce values in $[0.6, 1]$.
        Two times the standard deviation of the mean score is marked with vertical bars.
        All methods were run until computational time got prohibitive (ECP and MNWBS) or 32 GB of memory was insufficient for computation (KCP, MultiRank, and changekNN).
    }
\end{figure}

The change in mean method does not perform well for either setup.
Our method \texttt{changeforest} performs best, except for very small sample sizes.
For the dry beans simulation setup when $n / K < 20$, ECP performs best, but not much better than 0.8.
For the Dirichlet simulation setup for $K = 20, n = 354$ and $K = 80, n = 1000$, MultiRank performs best.
Surprisingly, KCP performs very badly for the dry beans-based simulation, not selecting any change points at all.
This might be due to the higher number of change points compared to the main simulations or the added Gaussian noise.

We display the average computational time of the different methods in Figure~\ref{fig:evolution_time}, together with dashed grey lines corresponding to linear and quadratic time complexity.
Again, the parametric change in mean method is very fast.
For both change in mean and \texttt{changeforest}, time complexity appears to scale approximately linear in the number of observations.
Meanwhile, changekNN, MultiRank, KCP, and MNWBS scale approximately quadratic, while ECP appears to have worse than quadratic time complexity. \texttt{changeforest} is the quickest nonparametric method for $n \gtrsim 5000$.
\begin{figure}[h]
    \centering
    \includegraphics[width = \textwidth]{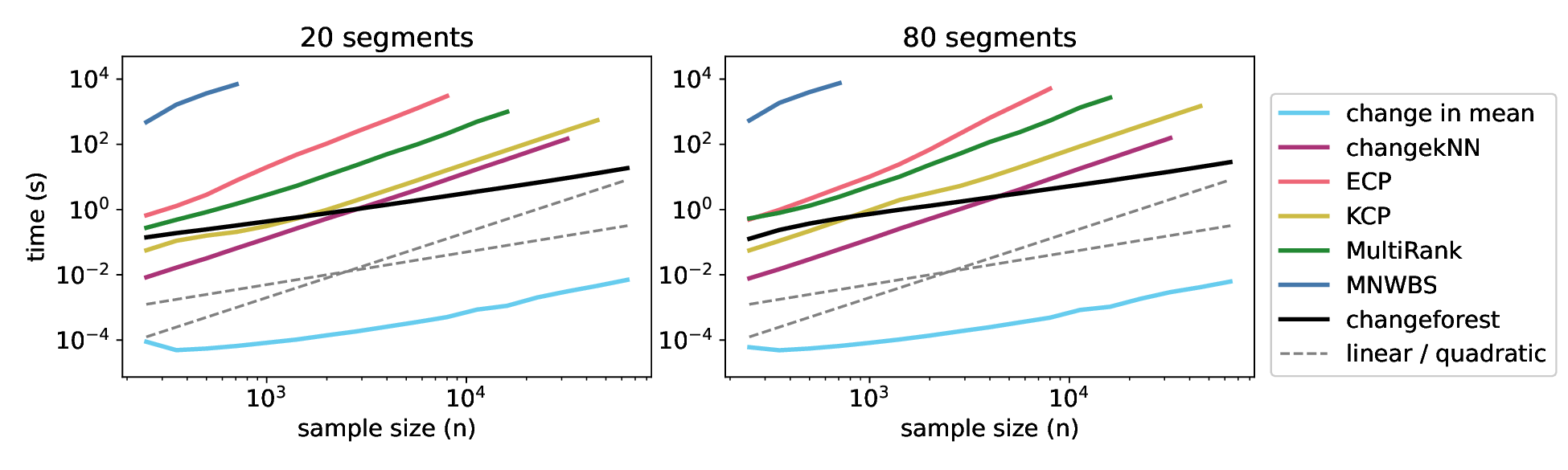}
    \vspace{-0.2cm}
    \caption{
        \label{fig:evolution_time}
        Average time for change point detection on 8 Intel Xeon 2.3 GHz cores with 4 GB of RAM available per core for the Dirichlet simulation setup.
        All methods were run until computational time got prohibitive (ECP and MNWBS) or 32 GB of memory was insufficient (KCP, MultiRank, and changekNN).
    }
\end{figure}

\subsection{Dependence on hyperparameters}\label{sec:tuning}
We investigate how the hyperparameters of random forests affect \texttt{changeforest}'s performance.
For this, we vary the three most important parameters for random forests: the number of trees from $20$ to $100$ (default) to $500$, the maximal depth of individual trees from $2$ to $8$ (default) to infinity, and the number of features considered at each split (mtry) from $1$ to $\sqrt{d}$ (default, rounded down) to $d$.
The average adjusted Rand indices for each set of parameters for the simulation setups introduced in Section~\ref{sec:simulations_setups} are in Table \ref{tab:tuning_parameters} in the Appendix.
The average computational times are in Table \ref{tab:tuning_parameters_time}.
We summarize the results here.

Choosing even very extreme parameters, such as $\textrm{number of trees}=20$, $\textrm{maximal depth}=2$, and $\textrm{mtry}=1,$ yields acceptable change point estimation performance for simple simulation setups.
With these parameters, \texttt{changeforest} performs on average better than ECP, MultiRank, MNWBS, and changekNN, our method paired with a $k$-nearest neighbor classifier.
It even outperforms the default configuration (100, 8, $\sqrt{d}$) for the iris, breast cancer, and wine simulation setups, possibly due to selecting fewer false positives because of higher noise in the predictions.
On the other hand, increasing the number of trees from 100 to 500 and splitting individual trees to purity does not significantly increase the performance compared to the default values while increasing computational cost fivefold.

This low dependence on the choice of tuning parameters is in contrast to KCP, where the bandwidth of the Gaussian kernel has to be fine-tuned for optimal performance.
Table \ref{tab:tuning_parameters_kcp} in the Appendix shows the estimation performance of KCP for different kernels and bandwidths.

\subsection{False positive rates}\label{sec:false_positive_rates}
We empirically analyze our model selection procedure.
For this, we select the largest class of each data set from Section~\ref{sec:simulations_setups}, shuffle the observations before applying the change point detection methods, and collect the percentage out of 2500 simulations where at least one change point was detected.
See Table \ref{tab:false_positive} for the results.
\begin{table}[h]
    \caption{
        \label{tab:false_positive}
        Percentage of simulations where at least one change point was detected in a homogeneous data set.
        For each data set and method, 2500 simulations were performed.
    }

    \makebox[\textwidth][c]{
        \begin{tabular}{lrrrrrrrrrr}
            data set & CIM & CIC & Dirichlet & iris & glass & breast & abalone & wine & dry &  worst \\
                    &     &     &           &      &       & cancer &         &      & beans &   \\
            \midrule
            change in mean      &          0.12 &          0.12 &              100.0 &          27.04 &           91.44 &                    3.44 &              7.68 &           1.20 &                1.40 &  100.00 \\ 
            changekNN           &          2.44 &          2.44 &                1.80 &           3.64 &            2.40 &                   29.80 &              3.28 &           6.80 &                2.04 &  29.80 \\
            ECP                 &          5.76 &          6.20 &                5.08 &           6.80 &            5.72 &                    4.76 &              4.76 &           5.16 &                5.16 &  6.80 \\
            KCP           &          0.00 &          0.00 &                0.00 &           3.04 &            0.32 &                    5.24 &             11.32 &           0.08 &                0.04 &  11.32 \\
            MultiRank           &          0.20 &          0.20 &               12.08 &          93.64 &           90.44 &                  100.0 &            100.0 &         100.0 &                     & 100.00 \\
            MNWBS             &   95.08        & 95.08           &    79.36        &  95.16        &   77.68           &        98.20              &&&&    98.20\\
            \texttt{changeforest} &          3.36 &          3.36 &                3.76 &           5.00 &            3.80 &                    3.80 &              3.00 &           3.52 &                2.48 &  5.00 \\
        \end{tabular}
    }
\end{table}

As our model selection procedure is not a true permutation test, we observe more false positives than expected, given the approximate significance threshold of $0.02$ for both \texttt{changeforest} and the $k$-nearest neighbor method changekNN.
However, \texttt{changeforest} has a similar false positive rate across all simulation setups.
The methods based on heuristics (KCP and MNWBS) or asymptotic theory (MultiRank, change in mean) have varying degrees of false positive rates by data set.
ChangekNN, which uses the same model selection procedure as \texttt{changeforest}, has a false positive rate of almost 30\% for the breast cancer data set.

\section{Future work}
We comment on some possible avenues for future work.

\paragraph*{Time series featuring serial correlations}
The setting outlined in Section~\ref{sec:non_parametric_methodology} requires independent observations.
In practice, time series data often exhibit serial correlation and the change in distribution may only manifest through a change in the structure of this serial correlation.
Even though the model assumptions are violated, our algorithm can be applied to such data sets by augmenting the data with lagged observations.
It would be interesting to compare the performance of \texttt{changeforest} when applied to such augmented time series with the performance of change point detection methods specifically designed to detect changes in the autocorrelation structure.

\paragraph*{Extending the theoretical guarantees to multiple change points}
As discussed in Section~\ref{sec:theoretical_guarantees}, the two-step search is not guaranteed to find a true change point in a multiple change point scenario due to possible canceling out effects (see also Figure \ref{fig:two_step_search}).
One possible remedy would be to pair the two-step search with wild binary segmentation \citep{fryzlewicz2014wild} or seeded binary segmentation \citep{kovacs2020seeded}.
This would help in scenarios with very short segments and equal or similar distributions in surrounding segments but would require changes to the model selection procedure presented to avoid overfitting.
As we already observe very good empirical performance with binary segmentation, we did not pursue this further.

\section{Conclusion}\label{sec:discussion}

Our proposed \texttt{changeforest} algorithm uses random forests to efficiently and accurately perform multiple change point detection across a wide range of data set types.
Motivated by parametric change point detection methodology, the algorithm optimizes a nonparametric classifier-based log-likelihood ratio.
When coupled with random forests, it performs very well or better than competitors for all simulation setups studied.
The computation time of \texttt{changeforest} scales almost linearly with the number of observations, enabling nonparametric change point detection even for very large data sets.
The \texttt{changeforest} software package is available to \texttt{R}, Python, and Rust users.
See \url{github.com/mlondschien/changeforest} for installation instructions.

\acks{Malte Londschien is supported by the ETH Foundations of Data Science and the ETH AI Center.
Solt Kov\'acs and Peter B\"uhlmann are supported by the European Research Council (ERC) under the European Union’s Horizon 2020 research and innovation programme (Grant agreement No. 786461
CausalStats - ERC-2017-ADG).
}

\contrib{S.K.\ conceived and co-supervised the project.
    M.L.\ and S.K.\ developed the methodology.
    M.L.\ designed and implemented the changeforest software package, designed, implemented, and ran the simulations, developed the theory, and wrote the manuscript.
    P.B.\ co-supervised the project and provided feedback and ideas during the development.
    S.K.\ and P.B.\ provided feedback on the manuscript.
}

\bibliography{bib}

\appendix
\clearpage
\section{Proofs}
\label{appendix:proofs}

\lemmaoptimalloglikelihood*
\begin{proof}
    Let $u<v$ and define $\kappa(\uv) \coloneqq \{k \in\{1, \ldots, K\} \mid \seg{\alpha^0_{k \shortminus 1}}{\alpha^0_k} \cap \uv \neq \emptyset\}$.
    If there exists some $k$ such that $\alpha_{k \shortminus 1}^0 \leq u < v \leq \alpha_k^0$, then $\kappa(\uv) = \{ k \}$ and $p_\uv = p_\seg{\alpha_{k \shortminus 1}^0}{\alpha_k^0} = p_k$. Consequently, as the $X_{u+1}, \ldots, X_{v-1}$ are $\BP_k$-distributed, we have
    \begin{equation*}
        \label{eq_appendix:dkl1}
        \BE\left[ \sum_{i = u+1}^v \log(\frac{p_\uv(X_i)}{p_\seg{0}{n}(X_i)}) \right]
        = (v - u) \DKL{\BP_k}{\BP_\seg{0}{n}}.
    \end{equation*}
    This implies that, for any segmentation $\alpha$ containing $\alpha^0$ and thus $|\kappa(\segk)| = 1$ for all $k = 1, \ldots, K$,
    \begin{equation*}
    \BE\left[ G\left( (X_i)_{i = 1}^n \mid \alpha, p^\infty \right)\right]
    = \sum_{k = 1}^K (\alpha^0_k - \alpha^0_{k \shortminus 1}) \DKL{\BP_k}{\BP_\seg{0}{n}}.
    \end{equation*}
    Now assume that $\kappa(\uv)$ contains more than a single element.
    Set $\pi_k(\uv) \coloneqq \frac{|\seg{\alpha^0_{k \shortminus 1}}{\alpha^0_k} \cap \uv|}{v - u}$.
    Then $\BP_\uv = \sum_{k \in\kappa(\uv)} \pi_k(\uv) \BP_k$ and
    \begin{align*}
        \BE\left[ \sum_{i = u+1}^v \log(\frac{p_\uv(X_i)}{p_\seg{0}{n}(X_i)}) \right] &= 
        (v - u)\BE_{X \sim \BP_\uv}
        \left[ \log(\frac{p_\uv(X)}{p_\seg{0}{n}(X)}) \right] \\
        &= (v - u) %
        \BE_{X \sim \BP_\uv}
        \left[ \log(
            \frac{\sum_{k \in \kappa(\uv)} \pi_k(\uv) p_k(X)}{p_\seg{0}{n}(X)}
        ) \right] \\
        &< (v - u) \sum_{k \in \kappa(\uv)} \pi_k(\uv) \BE_{X \sim \BP_k}\left[ \log(\frac{p_k(X)}{p_\seg{0}{n}(X)})   \right]\\
        &= (v - u) \sum_{k \in \kappa(\uv)} \pi_k(\uv) \DKL{\BP_k}{\BP_\seg{0}{n}}\, .
    \end{align*}
    In the third line, we used the fact that $x \mapsto x\log(x)$ is strictly convex.
    Consider any segmentation $\alpha$ that does not contain $\alpha^0$.
    Then there exists at least one segment $\seg{\alpha_{k'\shortminus 1}}{\alpha_{k'}}$ such that $\kappa(\seg{\alpha_{k'\shortminus 1}}{\alpha_{k'}})$ contains more than one element and thus
    \begin{align*}
        \BE \left[ G \left( (X_i)_{i = 1}^n \mid \alpha, p^\infty \right) \right]
        &=
        \sum_{k=1}^K \BE \left[ \sum_{i = \alpha_{k\shortminus 1}+1}^{\alpha_k} \log(\frac{p_\segk(X_i)}{p_\seg{0}{n}(X_i)}) \right] \\
        &< \sum_{k=1}^K (\alpha_k - \alpha_{k \shortminus 1}) \sum_{l \in \kappa(\segk)} \pi_l(\segk) \DKL{\BP_l}{\BP_\seg{0}{n}} \\
        &= \sum_{k = 1}^K (\alpha^0_k - \alpha^0_{k \shortminus 1}) \DKL{\BP_k}{\BP_\seg{0}{n}}\, ,
    \end{align*}
    applying the above inequality with $u=\alpha_{k' \shortminus 1}$ and $v=\alpha_{k'}$.
\end{proof}

\lemmaconvexity*
\begin{proof}
    Let $s \in \seg{u}{v-1} \backslash \alpha^0$.
    We show that the expected classifier log-likelihood ratio
    \begin{equation*}
        G(s) := \BE \left[ G \left( (X_i)_{i = u+1}^v \mid \{u, s, v\}, p^\infty \right) \right]
    \end{equation*}
    is convex at $s$, that is, $G(s+1) - 2 G(s) + G(s-1) > 0$. By definition
    \begin{equation*}
        G(s) = \BE \left[ \sum_{i = u+1}^s \log(\frac{p_\us(X_i)}{p_\uv(X_i)}) + \sum_{i = s+ 1}^v \log(\frac{p_\sv(X_i)}{p_\uv(X_i)}) \right]
    \end{equation*}
    and
    \begin{align*}
        G(s+1) - G(s)
        &=
        \BE \left[
            \sum_{i = u + 1}^{s+1} \log(\frac{p_\seg{u}{s+1}(X_i)}{p_\us(X_i)}) +
            \sum_{i = s + 2}^v \log(\frac{p_\seg{s+1}{v}(X_i)}{p_\sv(X_i)}) +
            \log(\frac{p_\us(X_{s+1})}{p_\sv(X_{s+1})})
        \right] \\
        &= (s - u + 1) \DKL{\BP_\seg{u}{s+1}}{\BP_\us} +
        (v - s - 1) \DKL{\BP_\seg{s+1}{v}}{\BP_\sv}\ + \\
        &\hphantom{\hspace{7cm}} \BE \left[ \log(\frac{p_\us(X_{s+1})}{p_\sv(X_{s+1})}) \right]
    \end{align*}
    Similarly,
    \begin{align*}
        G(s-1) - G(s)
        &=
        \BE \left[
            \sum_{i = u + 1}^{s - 1} \log(\frac{p_\seg{u}{s-1}(X_i)}{p_\us(X_i)}) +
            \sum_{i = s}^v \log(\frac{p_\seg{s-1}{v}(X_i)}{p_\sv(X_i)}) +
            \log(\frac{p_\sv(X_{s})}{p_\us(X_{s})})
        \right] \\
        &= (s - u - 1) \DKL{\BP_\seg{u}{s-1}}{\BP_\us} +
        (v - s + 1) \DKL{\BP_\seg{s-1}{v}}{\BP_\sv} \ + \\
        &\hphantom{\hspace{7cm}} \BE \left[ \log(\frac{p_\sv(X_s)}{p_\us(X_s)}) \right] 
    \end{align*}
    As $s\notin \alpha^0$, the $X_s, X_{s+1}$ are i.i.d.\ and thus
    \begin{equation*}
        \BE \left[ \log(\frac{p_\us(X_{s+1})}{p_\sv(X_{s+1})}) +
         \log(\frac{p_\sv(X_s)}{p_\us(X_s)}) \right]
        =
        \BE \left[ \log(\frac{p_\us(X_{s+1})}{p_\us(X_{s})})  +
        \log(\frac{p_\sv(X_s)}{p_\sv(X_{s+1})}) \right]
        = 0.
    \end{equation*}
    Combining everything,
    \begin{align*}
        G(s + 1) - &2 G(s) + G(s - 1) = (G(s+1) - G(s)) + (G(s-1) - G(s)) \\
        &= (s - u + 1) \DKL{\BP_\seg{u}{s+1}}{\BP_\us}\ +
        (v - s - 1) \DKL{\BP_\seg{s+1}{v}}{\BP_\sv}\ + \\
        &\hphantom{\hspace{1cm} } (s - u - 1) \DKL{\BP_\seg{u}{s-1}}{\BP_\us}\ +
        (v - s + 1) \DKL{\BP_\seg{s-1}{v}}{\BP_\sv}.
    \end{align*}
    This is nonnegative, as the Kullback-Leibler divergence between two distributions is nonnegative and as $u+1 \leq s \leq v-1$.
    Furthermore, as $s - u + 1 \geq 1$ and $v - s + 1 \geq 1$, this term is positive if $\BP_\seg{u}{s+1} \neq \BP_\us$ or $\BP_\seg{s - 1}{v} \neq \BP_\sv$.
    Choose $k$ such that $s \in \seg{\alpha_{k \shortminus 1}^0}{\alpha_k^0}$.    
    Then, as $s\notin \alpha^0$, also $s + 1 \in \seg{\alpha_{k \shortminus 1}^0}{\alpha_k^0}$.
    Then $\BP_\seg{u}{s+1} \neq \BP_\us$ or $\BP_\seg{s - 1}{v} \neq \BP_\sv$ if and only if $\BP_\seg{u}{\alpha_{k \shortminus 1}^0} \neq \BP_\seg{\alpha^0_{k\shortminus 1}}{\alpha^0_k}$ or $\BP_\seg{\alpha^0_{k\shortminus 1}}{\alpha^0_k} \neq \BP_\seg{\alpha_k^0}{v}$.
\end{proof}

\lemmapiecewiselinear*
\begin{proof}
Write
\begin{align*}
    G(s) &\coloneqq \sum_{i = u + 1}^{s} \log( \frac{v - u}{s^{(0)} - u} p_{\{u, s^{(0)}, v\}}^\infty(X_i)_1) + \sum_{i = s + 1}^v \log(\frac{v - u}{v - s^{(0)}} p_{\{u, s^{(0)}, v\}}^\infty(X_i)_2) \\
    &= \sum_{i=u+1}^s \log( \frac{p_\seg{u}{s^{(0)}}(X_i)}{p_\uv(X_i)} ) + \sum_{i=s+1}^v \log( \frac{p_\seg{s^{(0)}}{v}(X_i)}{p_\uv(X_i)} )
\end{align*}
and define
\begin{equation*}
    U_j \coloneqq G(j) - G(j - 1) = 
    \log( \frac{p_\seg{u}{s^{(0)}}(X_j)}{p_\seg{s^{(0)}}{v}(X_j)}),
\end{equation*}
such that $G(s) = G(u) + \sum_{j = u+1}^s U_j$.
For any segment $\seg{\alpha_{k \shortminus 1}^0}{\alpha_k}$ intersecting the segment $\uv$ and any $s \in \seg{\alpha_{k \shortminus 1}^0}{\alpha_k}$, the $U_{\alpha_{k \shortminus 1}^0 + 1}, \ldots, U_{\alpha_k}$ are i.i.d.
Consequently, $\BE[G(s)] = \BE[G(\alpha_{k \shortminus 1}^0)] + (s - \alpha_{k \shortminus 1}^0) \BE[U_{\alpha_{k \shortminus 1}^0 + 1}]$, proving that $\BE[G(s)]$ is piecewise linear between the change points in $\alpha^0$.

Now assume that there is a single change point $a^0$ in $\uv$.
Assume without loss of generality that $s^{(0)} \geq a^0$ such that $p_\seg{u}{s^{(0)}} = \frac{a^0 - u}{s^{(0)} - u} p_\seg{u}{a^0} + \frac{s^{(0)} - a^0}{s^{(0)} - u}p_\seg{a^0}{v}$ and $p_\seg{s^{(0)}}{v} = p_\seg{a^0}{v}$.
Then, for $j \leq a^0$,
\begin{align*}
    \BE[U_j] &= \BE_{\BP_\seg{u}{a^0}}
    \left[ \log(
    \frac{p_\seg{u}{s^{(0)}}(X_j)}{p_\seg{s^{(0)}}{v}(X_j)}
    ) \right]\\
    &= \frac{s^{(0)} - u}{a^0 - u}
    \BE_{\BP_\seg{u}{s^{(0)}}}\left[
    \log( \frac{p_\seg{u}{s^{(0)}}(X_j)}{p_\seg{s^{(0)}}{v}(X_j)}) 
    \right] - 
    \frac{s^{(0)}-a^0}{a^0-u}
    \BE_{\BP_\seg{a^0}{v}}    \left[ \log(
    \frac{p_\seg{u}{s^{(0)}}(X_j)}{p_\seg{s^{(0)}}{v}(X_j)}
    )\right]\\
    &= \frac{s^{(0)} - u}{a^0 - u} \DKL{\BP_\seg{u}{s^{(0)}}}{\BP_\seg{a^0}{v}} + 
    \frac{s^{(0)} - a^0}{a^0 - u} \DKL{\BP_\seg{s^{(0)}}{v}}{\BP_\seg{u}{s^{(0)}}} > 0.
\end{align*}
Similarly, for $j > a^0$
\begin{equation*}
    \BE[U_s] = \BE_{\BP_\seg{s^{(0)}}{v}}
    \left[ \log(
    \frac{p_\seg{u}{s^{(0)}}(X_j)}{p_\seg{s^{(0)}}{v}(X_j)}
    ) \right] = -  \DKL{\BP_{(a^{(0)}, v]}}{\BP_\seg{u}{s^{(0)}}} < 0,
\end{equation*}
which shows that $s \mapsto \BE[G(s)] = \BE[G(u)] + \sum_{i=u+1}^s \BE[U_i]$ has a unique maximum at $s = a^0$.
\end{proof}

\begin{restatable}{lemma_}{kl} \label{lem:kl}
    Let $\BP_1, \BP_2$ be probability measures with corresponding densities $p_1, p_2$.
    Let $\delta_1, \delta_2 \in [0,1]$ with $\delta_1 \neq \delta_2$.
    Define $\BQ(\delta) \coloneqq (1 - \delta)\BP_1 + \delta \BP_2$ with density $q(\delta) = (1 - \delta)p_1 + \delta p_2$.
    Then
    \begin{align}
        \BE_{\BP_1}\left[\log( \frac{q(\delta_1)}{q(\delta_2)})\right] =\ &
        \frac{\delta_2}{\delta_2 - \delta_1} \DKL{\BQ(\delta_1)}{\BQ(\delta_2)}\ +\\
        &\frac{\delta_1}{\delta_2 - \delta_1} \DKL{\BQ(\delta_2)}{\BQ(\delta_1)}.
    \end{align}
\end{restatable}

\begin{proof}
    The Kullback Leibler divergence between distribution $Q_1, Q_2$ with densities $q_1, q_2$ is $\DKL{Q_1}{Q_2} = \BE_{Q_1}[\log(\frac{q_1(X)}{q_2(X)})]$. Thus
    \begin{equation*}
        \DKL{\BQ(\delta_1)}{\BQ(\delta_2)} = (1 - \delta_1) \BE_{\BP_1}\left[\log( \frac{q(\delta_1)}{q(\delta_2)})\right] + \delta_1 \BE_{\BP_2}\left[\log( \frac{q(\delta_1)}{q(\delta_2)})\right]
    \end{equation*}
    and
    \begin{equation*}
        \DKL{\BQ(\delta_2)}{\BQ(\delta_1)} =
        (1 -\delta_2) \BE_{\BP_1}\left[\log( \frac{q(\delta_2)}{q(\delta_1)})\right] +
        \delta_2 \BE_{\BP_2}\left[\log( \frac{q(\delta_2)}{q(\delta_1)})\right]
    \end{equation*}
    such that
    \begin{multline*}
       \frac{\delta_2}{\delta_2 - \delta_1} \DKL{\BQ(\delta_1)}{\BQ(\delta_2)} +
        \frac{\delta_1}{\delta_2 - \delta_1} \DKL{\BQ(\delta_2)}{\BQ(\delta_1)} \\
        = \left(\frac{\delta_2}{\delta_2 - \delta_1} (1 - \delta_1) - \frac{\delta_1}{\delta_2 - \delta_1}(1 - \delta_2)\right) \BE_{\BP_1}\left[\log( \frac{q(\delta_1)}{q(\delta_2)})\right] +\\
        \left(\frac{\delta_2}{\delta_2 - \delta_1}\delta_1 - \frac{\delta_1}{\delta_2 - \delta_1}\delta_2 \right) \BE_{\BP_2}\left[\log( \frac{q(\delta_1)}{q(\delta_2)})\right]
        = \BE_{\BP_1}\left[ \log(\frac{q(\delta_1)}{q(\delta_2)})\right].
    \end{multline*}
\end{proof}

\begin{lemma}
    \label{lem:union_bound}
    Let $Z_1, \ldots, Z_n$ be independent random variables bounded from below by $a$ and from above by $b$. Let $1 \leq K \leq n$. Then,
    $$
    \BP \left[ \bigcup_{k=1, \ldots, K} \left\{ \sum_{i=1}^k (Z_i - \BE[Z_i]) \geq (b - a) \sqrt{k \log(n)} \right\} \right] \leq \frac{K}{n^2}.
    $$
\end{lemma}

\begin{proof}
    Define $\tilde Z_i := Z_i - \BE[Z_i]$ such that $a_i := a - \BE[Z_i] \leq Z_i \leq b - \BE[Z_i] := b_i$, with $b_i - a_i = b - a$ for all $i = 1, \ldots, n$.
    By Hoeffding's theorem, for each $k = 1, \ldots, K$ and $\lambda > 0$, we have
    $$
    \BP \left[ \sum_{i=1}^k \tilde Z_i \geq \lambda \right] \leq \exp( - \frac{2 \lambda^2}{k (b_i - a_i)^2}).
    $$
    Setting $\lambda = (b - a) \sqrt{k \log(n)}$ and taking a union bound yields
    $$
    \BP \left[ \bigcup_{k=1, \ldots, K} \left\{ \sum_{i=1}^k (Z_i - \BE[Z_i]) \geq (b - a) \sqrt{k \log(n) } \right\} \right] \leq K \exp( - \frac{2 k (b - a)^2 \log(n)}{k (b - a)^2}) = \frac{K}{n^2}.
    $$    
\end{proof}

\finitesample*

\begin{proof}
    Recall that $\log_\eta(x) = \log(\eta + (1 - \eta) x)$ and write
    $$
    G^{s, \eta}(t, p) = \sum_{i = 1}^t \log_\eta\left(\frac{n}{s} p_s(X_i)_1\right) + \sum_{i = t+1}^n \log_\eta\left(\frac{n}{n-s} p_s(X_i)_2\right)
    $$
    for $t = 0, \ldots, n$ and any classifier $p$. Recall that $p^\infty$ is the Bayes classifier with $p_s^\infty = \left(\frac{s}{n} \frac{p_{(0, s]}}{p_{(0, n]}}, \frac{n-s}{n} \frac{p_{(s, n]}}{p_{(0, n]}}\right)^T$.
    We proceed in four steps. \\
    
    \noindent (i) First, we show that, for any $t = 1, \ldots, n$,
    $$
    \BE[ G^{s, \eta}(a^0, p^\infty) - G^{s, \eta}(t, p^\infty)] \geq  2 |t - a^0| \delta (1 - \eta)^2 \dTV{\BP_1}{\BP_2}^2.
    $$
    
    \noindent (ii) Second, we show that
    \begin{multline*}
    \BP\Bigg[ \bigcap_{t = 1, \ldots n} \bigg\{ \left[G^{s, \eta}(a^0, p^\infty) - G^{s, \eta}(t, p^\infty)\right] - \BE \left[G^{s, \eta}(a^0, p^\infty) - G^{s, \eta}(t, p^\infty)\right] \\ > - 2 \log(\frac{1}{\rho (1 - \rho)} \frac{1}{\eta^2}) \sqrt{|t - a^0| \log(n)} \bigg\} \Bigg] \geq 1 - \frac{1}{n} \overset{n \to\infty}{\longrightarrow} 1.
    \end{multline*}

    \noindent (iii) Next, we argue by $\frac{1}{\eta}$-Lipschitz continuity of $\log_\eta$ and consistency of $\hat p$ that, for any $\varepsilon >0$,
    \begin{equation*}
        \BP \left[ \bigcap_{t = 1, \ldots, n} \Big\{ [G^{s, \eta}(a^0, \hat p) - G^{s, \eta}(t, \hat p)] - [G^{s, \eta}(a^0, p^\infty) - G^{s, \eta}(t, p^\infty)] >- |t - a^0| \varepsilon  \Big\} \right] \overset{n \to \infty}{\longrightarrow} 1.
    \end{equation*}
    
    \noindent (iv) Combined, these steps imply that, if $|t - a^0| \geq \frac{C \log(n)}{\delta^2 \dTV{\BP_1}{\BP_2}^4}$, then, with high probability, $G^{s, \eta}(a^0, \hat p) > G^{s, \eta}(t, \hat p)$ and thus $|\hat a - a^0| < \frac{C \log(n)}{\delta^2 \dTV{\BP_1}{\BP_2}^4}$.
    
    \paragraph*{Step 1:}
    For $i = 1, \ldots, n$, define
    $$
    U^{s, \eta}(p)_i := G^{s, \eta}(i, p) - G^{s, \eta}(i-1, p) = \log_\eta \left(\frac{n}{s} p_s(X_i)_1 \right) - \log_\eta \left(\frac{n}{n-s} p_s(X_i)_2 \right).
    $$
    We show that
    $$
    | \BE[ U^{s, \eta}(p^\infty)_i] |  \geq 2 (1 - \eta)\delta d_\textrm{TV}(\BP_1, \BP_2)^2,
    $$
    with a positive sign for $i \leq a^0$ and a negative sign for $i > a^0$.
    
    By symmetry, we can assume, without loss of generality, that $s \geq a^0$.
    Then, $\BP_{(0, s]} = \frac{a^0}{s} \BP_1 + \frac{s - a^0}{s} \BP_2$ and $\BP_{(s, n]} = \BP_2$ with densities $p_{(0, s]} = \frac{a^0}{s} p_1 + \frac{s - a^0}{s} p_2$ and $p_{(s, n]} = p_2$.
    Furthermore, $\BP_{(0, n]} = \frac{a^0}{n} \BP_1 + \frac{n - a^0}{n} \BP_2$ with density $p_{(0, n]} = \frac{a^0}{n} p_1 + \frac{n - a^0}{n} p_2$.
    For each $i = 1, \ldots n$, we have
    
    \begin{align*}
    U^{s, \eta}(p^\infty)_i &= \log \left(
        \eta + (1 - \eta) \frac{p_{(0, s]}(X_i)}{p_{(0, n]}(X_i)}
    \right) - \log \left(
        \eta + (1 - \eta) \frac{p_{(s, n]}(X_i)}{p_{(0, n]}(X_i)}
    \right) \\
    &= \log \left(
        \frac{\eta p_{(0, n]}(X_i) + (1 - \eta) p_{(0, s]}(X_i)}{p_{(0, n]}(X_i)}
    \right) - \log \left(
        \frac{\eta p_{(0, n]}(X_i) + (1 - \eta) p_{(s, n]}(X_i)}{p_{(0, n]}(X_i)}
    \right) \\
    &= \log \left(
        \frac
        {
            \eta(\frac{a^0}{n} p_1(X_i) + \frac{n - a^0}{n} p_2(X_i)) + (1 - \eta) (\frac{a^0}{s} p_1 + \frac{s - a^0}{s} p_2)(X_i)
        }{
            \eta(\frac{a^0}{n} p_1(X_i) + \frac{n - a^0}{n} p_2(X_i)) + (1 - \eta) p_2(X_i)
        }
    \right) \\
    &= \log \left(
        \frac
        {\frac{a^0}{s}(1 - \eta \frac{n - s}{n}) p_1(X_i) + (1 - \frac{a^0}{s}(1 - \eta \frac{n - s}{n}))p_2(X_i)}
        {\eta \frac{a^0}{n} p_1(X_i) + (1 - \eta \frac{a^0}{n}) p_2(X_i)}
    \right).
    \end{align*}
    Let $\BQ(\delta) = (1 - \delta) \BP_1 + \delta \BP_2$.
    Then, by Lemma \ref{lem:kl} with $\delta_1 = 1 - \frac{a^0}{s}(1 - \eta \frac{n - s}{n})$, $\delta_2 = 1 - \eta \frac{a^0}{n}$, and thus $\delta_2 - \delta_1 = (1 - \eta) \frac{a^0}{s}$,
    
    \begin{align*}
    \BE_{X_i \sim \BP_1}[U^{s, \eta}(p^\infty)_i]
    = \frac{1 - \eta \frac{a^0}{n}}{(1 - \eta) \frac{a^0}{s}} \DKL{\BQ(\delta_1)}{\BQ(\delta_2)} + \frac{1 - \frac{a^0}{s}(1 - \eta \frac{n-s}{n})}{(1 - \eta) \frac{a^0}{s}} \DKL{\BQ(\delta_2)}{\BQ(\delta_1)}.
    \end{align*}
    By Pinsker's inequality, $\DKL{\BQ(\delta_1)}{\BQ(\delta_2)} \geq 2 \dTV{\BQ(\delta_1)}{\BQ(\delta_2)}^2$.
    Furthermore, it holds that $ \dTV{\BQ(\delta_1)}{\BQ(\delta_2)} = | \delta_2 - \delta_1 | \dTV{\BP_1}{\BP_2}$.
    Thus,
    
    \begin{align*}
        \BE_{X_i \sim \BP_1}[U^{s, \eta}(p^\infty)_i]
        &\geq 2 (1 - \eta) \frac{a^0}{s} \left( 1 - \eta \frac{a^0}{n} + 1 - \frac{a^0}{s}(1 - \eta \frac{n-s}{n})\right) \dTV{\BP_1}{\BP_2}^2\\
        &= 2 (1 - \eta) \frac{a^0}{s} \left( 2 - \frac{a^0}{s} (1 - \eta \frac{n-2 s}{n} )\right) \dTV{\BP_1}{\BP_2}^2\\
        &\geq 2 (1 - \eta)^2 \delta \dTV{\BP_1}{\BP_2}^2.
    \end{align*}
    Similarly,
    \begin{align*}
        - \BE_{X_i \sim \BP_2}[U^{s, \eta}(p^\infty)_i] &\geq 2 (1 - \eta)^2 \delta \dTV{\BP_1}{\BP_2}^2,
    \end{align*}
    which proves the claim.

    \paragraph*{Step 2:}
    Note that $\log_\eta(x)$ is bounded from below by $\log(\eta)$ and from above by $\log(x)$ for all $x \geq 1$.
    Consequently, the $U^{s, \eta}(p)_i = \log_\eta(\frac{n}{s}p_s(X_i)_1) - \log_\eta(\frac{n}{n - s}p_s(X_i)_2)$ are bounded from below by $a = - \log(\frac{1}{1 - \rho} \frac{1}{\eta})$ and from above by $b = \log(\frac{1}{\rho}\frac{1}{\eta})$ with $b - a = \log(\frac{1}{\eta^2}\frac{1}{\rho (1 - \rho)})$.
    Write
    $$
    \Delta^{s, \eta}(t) := \left[G^{s, \eta}(a^0, p^\infty) - G^{s, \eta}(t, p^\infty)\right] - \BE\left[G^{s, \eta}(a^0, p^\infty) - G^{s, \eta}(t, p^\infty)\right].
    $$
    Then, for $t \leq a^0$, 
    $$
    \Delta^{s, \eta}(t) = \sum_{i = t+1}^{a^0} (U^{s, \eta}(p^\infty)_i - \BE[U^{s, \eta}(p^\infty)_i]).
    $$
    We apply Lemma \ref{lem:union_bound} with $K = a^0$ and $Z_i = - U^{s, \eta}(p^\infty)_{a^0 - i + 1}$ for $i = 1, \ldots, K$.
    This yields
    \begin{align*}
        &\mspace{24mu} \BP\left[ \bigcup_{t = 1, \ldots, a^0} \left\{ \Delta^{s, \eta}(t) \leq - \log(\frac{1}{\eta^2} \frac{1}{\rho (1 - \rho)}) \sqrt{|t - a^0| \log(n)} \right\} \right] \\
        &= \BP\left[ \bigcup_{k = 1, \ldots, K} \left\{ \sum_{i = 1}^{k} Z_i - \BE[Z_i] \geq \log(\frac{1}{\eta^2} \frac{1}{\rho (1 - \rho)}) \sqrt{|k| \log(n)} \right\} \right] \leq \frac{K}{n^2} = \frac{a^0}{n^2}.
    \end{align*}
    Similarly, for $t \geq a^0 + 1$,
    $$
    \Delta^{s, \eta}(t) = \sum_{i = a^0 + 1}^{t} (U^{s, \eta}(p^\infty)_i - \BE[U^{s, \eta}(p^\infty)_i])
    $$
    and we can apply Lemma \ref{lem:union_bound} with $K = n - a^0$ and $Z_i = U_{a^0 + i}$ for $i = 1, \ldots, K$.
    This yields
    \begin{align*}
        &\mspace{24mu} \BP\left[ \bigcup_{t = a^0 + 1, \ldots, n} \left\{ \Delta^{s, \eta}(t) \leq - \log(\frac{1}{\eta^2} \frac{1}{\rho (1 - \rho)}) \sqrt{|t - a^0| \log(n)} \right\} \right] \\
        &= \BP\left[ \bigcup_{k = 1, \ldots, K} \left\{ \sum_{i = 1}^{k} Z_i - \BE[Z_i] \geq \log(\frac{1}{\eta^2} \frac{1}{\rho (1 - \rho)}) \sqrt{|k| \log(n)} \right\} \right] \leq \frac{K}{n^2} = \frac{n - a^0}{n^2}.
    \end{align*}
    Combining the two bounds, we obtain
    $$
    \BP\left[ \bigcup_{t = 1, \ldots, n} \left\{ \Delta(t) \leq - \log(\frac{1}{\eta^2} \frac{1}{\rho (1 - \rho)}) \sqrt{|t - a^0| \log(n)} \right\} \right] \leq \frac{n - a^0}{n^2} + \frac{a^0}{n^2} = \frac{1}{n}
    $$
    and thus
    $$
    \BP\left[ \bigcap_{t = 1, \ldots, n} \left\{ \Delta(t) > - \log(\frac{1}{\eta^2} \frac{1}{\rho (1 - \rho)}) \sqrt{|t - a^0| \log(n)} \right\} \right] \geq 1 -  \frac{1}{n}
    $$
    as desired.

    \paragraph*{Step 3:}
    Let $\varepsilon >0$. As $\log_\eta$ is $\frac{1}{\eta}$-Lipschitz, whenever
    $$
        \max_{t = 1, \ldots, n} |p_s^\infty(X_i)_1 - \hat p_s(X_i)_1| \leq \eta \rho (1 - \rho) \varepsilon 
    $$
    and thus also
    $$
    \max_{i=1, \ldots, n} | p_s^\infty(X_i)_2 - \hat p_s(X_i)_2 | = \max_{i=1, \ldots, n} | (1 - p_s^\infty(X_i)_1) - (1  - \hat p_s(X_i)_1) | \leq \eta \rho (1 - \rho) \varepsilon ,
    $$
    then
    $$
    \max_{i=1, \ldots, n} | U^{s, \eta}(p^\infty)_i - U^{s, \eta}(\hat p)_i | \leq \varepsilon ,
    $$
    and, for all $t = 1, \ldots, n$,
    \begin{equation*}
        \left| [G^{s, \eta}(a^0, \hat p) - G^{s, \eta}(t, \hat p)] - [G^{s, \eta}(a^0, p^\infty) - G^{s, \eta}(t, p^\infty)] \right| \leq |t - a^0| \varepsilon .
    \end{equation*}
    Thus, in particular,
    \begin{multline*}
        \BP\left[ \bigcap_{t=1, \ldots, n} \left\{ [G^{s, \eta}(a^0, \hat p) - G^{s, \eta}(t, \hat p)] - [G^{s, \eta}(a^0, p^\infty) - G^{s, \eta}(t, p^\infty)] \geq - |t - a^0| \varepsilon \right\} \right] \\ \geq
        \BP\left[ \bigcap_{t=1, \ldots, n} \left\{ \left| [G^{s, \eta}(a^0, \hat p) - G^{s, \eta}(t, \hat p)] - [G^{s, \eta}(a^0, p^\infty) - G^{s, \eta}(t, p^\infty)] \right| \leq |t - a^0| \varepsilon  \right\} \right] \\ \geq
        \BP\left[ \max_{t = 1, \ldots, n} |p_s^\infty(X_i)_1 - \hat p_s(X_i)_1| \leq \eta \rho (1 - \rho)\varepsilon  \right] \overset{n \to\infty}{\longrightarrow} 1
    \end{multline*}
    by the consistency assumption of $\hat p$, proving the claim.
    
    \paragraph*{Step 4:}
    We combine the results of steps 1, 2, and 3 to prove the claim. \\
    
    \noindent Let $t = 1, \ldots, n$. We bound $G(a^0, \hat p) - G(t, \hat p)$ from below with high probability.
    For this, we decompose
    \begin{multline*}
    G^{s, \eta}(a^0, \hat p) - G^{s, \eta}(t, \hat p) = 
    \underbrace{\BE[G^{s, \eta}(a^0, p^\infty) - G^{s, \eta}(t, p^\infty)]}_{A(t)}+ \\
    \underbrace{[G^{s, \eta}(a^0, p^\infty) - G^{s, \eta}(t, p^\infty)] - \BE[G^{s, \eta}(a^0, p^\infty) - G^{s, \eta}(t, p^\infty)]}_{B(t)} + \\
    \underbrace{[G^{s, \eta}(a^0, \hat p) - G^{s, \eta}(t, \hat p)] - [G^{s, \eta}(a^0, p^\infty) - G^{s, \eta}(t, p^\infty)]}_{C(t)}.
    \end{multline*}

    \noindent By step 1,
    $$
    \forall t = 1, \ldots, n: A(t) \geq 2 |t - a^0| \delta (1 - \eta)^2 \dTV{\BP_1}{\BP_2}^2.
    $$
    
    \noindent By step 2,
    \begin{equation*}
        \BP \left[ \forall t = 1, \ldots, n: B(t) > - \log(\frac{1}{\rho (1 - \rho)} \frac{1}{\eta^2}) \sqrt{|t - a^0| \log(n)} \right] \geq 1 - \frac{1}{n} \overset{n \to\infty}{\longrightarrow} 1.
    \end{equation*}
    
    \noindent By step 3, choosing $\varepsilon = \delta (1 - \eta)^2 \dTV{\BP_1}{\BP_2}^2$,
    \begin{equation*}
        \BP\left[ \forall t = 1, \ldots, n: C(t) \geq - |t - a^0| (1 - \eta)^2 \delta \dTV{\BP_1}{\BP_2}^2 \right] \overset{n\to\infty}{\longrightarrow} 1.
    \end{equation*}
    Thus
    \begin{multline*}
    \BP\Bigg[ \forall t = 1, \ldots, n: G^{s, \eta}(a^0, \hat p) - G^{s, \eta}(t, \hat p) = A(t) + B(t) + C(t) > \\ 
    |t - a^0| \delta (1 - \eta)^2 \dTV{\BP_1}{\BP_2}^2 - \log(\frac{1}{\rho (1 - \rho)} \frac{1}{\eta^2}) \sqrt{\log(n) |t - a^0|} \Bigg] \overset{n \to\infty}{\longrightarrow} 1
    \end{multline*}
    and
    \begin{equation*}
        \BP \left[ \forall t = 1, \ldots n \ \textrm{ s.t. } |t - a^0| \leq \frac{C \log(n)}{\delta^2 \dTV{\BP_1}{\BP_2}^4} : G^{s, \eta}(a^0, \hat p) > G^{s, \eta}(t, \hat p)\right] \overset{n \to \infty}{\longrightarrow} 1,
    \end{equation*}
    where $C = \frac{1}{(1 - \eta)^4} \log( \frac{n^2}{n-s} \frac{1}{\eta^2})^2$, implying
    $$
    \BP \left[ \left|\hat a - a^0 \right| > \frac{C \log(n)}{\delta^2 \dTV{\BP_1}{\BP_2}^4} \right] \overset{n \to \infty}{\longrightarrow} 0,
    $$
    as $\hat a \in \argmax_{t=1, \ldots, n} G^{s, \eta}(t, \hat p)$.

    \end{proof}

\newpage
\section{Figures}
\begin{figure}[H]
    \centering
    \includegraphics[width = \textwidth]{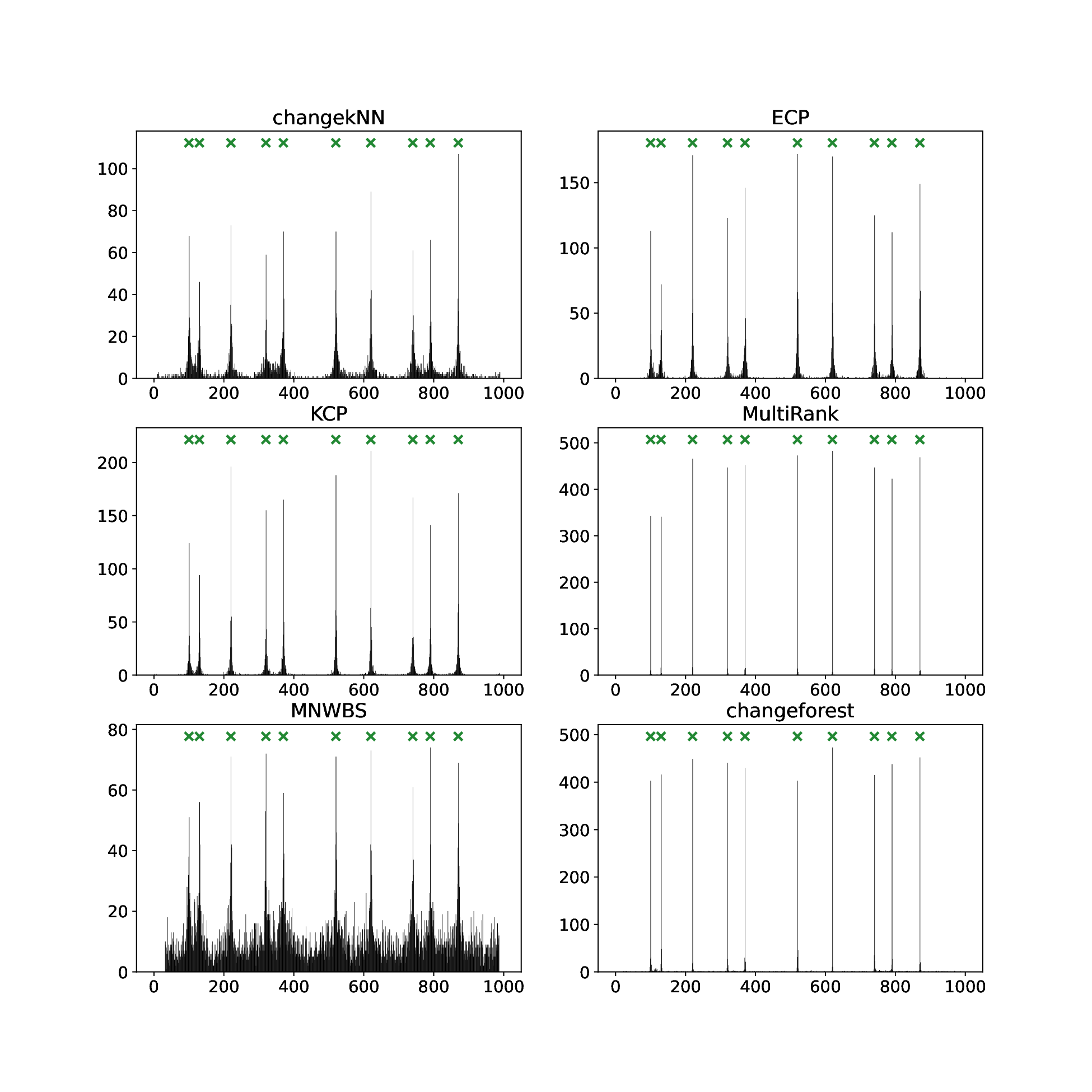}
    \vspace{-0.5cm}
    \caption{
        Histograms of cumulative change point estimates for 500 simulation runs on the Dirichlet simulation setup.
        The change in mean method is excluded as it did not produce any change point estimates for the Dirichlet simulation setup.
        The true change point locations are marked with green crosses.
        Note the different scales on the $y$-axis.
        \label{fig:histograms}
    }
\end{figure}

\newpage
\section{Tables}
\begin{table}[h]
    \caption{Adjusted Rand indices (ARI) and Hausdorff distances ($d_H$) for different segmentations. The true segmentation comes from the parametric change in mean (CIC, $\alpha^0=(0, 50, 100, 150)$)} and the nonparametric glass ($\alpha^0 = (0, 17, 46, 55, 68, 144, 214)$) setup.\label{tab:adj_rand_examples}
    \begin{small}
    \makebox[\textwidth][c]{
        \begin{tabular}{l l r r l}
            true segmentation & estimated segmentation & ARI & $d_H$ & comment \\ 
            \hline
            0, 50, 100, 150 & 0, 50, 100, 150 & 1.00 & 0.000 & perfect fit \\
            & 0, 52, 99, 150 & 0.94 & 0.013 & almost perfect fit \\
            & 0, 23, 50, 100, 150 & 0.87 & 0.153 & one extra change point \\
            & 0, 43, 87, 97, 150 & 0.75 & 0.087 & two extra change points \\
            & 0, 50, 150 & 0.57 & 0.333 & one missing change point \\
            & 0, 20, 70, 150 & 0.37 & 0.200 & random segmentation \\
            0, 17, 46, 55, 68, 144, 214 & 0, 17, 46, 55, 68, 144, 214 & 1.00 & 0.000 & perfect fit \\
            & 0, 15, 45, 55, 68, 142, 214 & 0.95 & 0.009 & almost perfect fit \\
            & 0, 17, 46, 55, 68, 80, 144, 214 & 0.91 & 0.056 & one extra change point \\
            & 0, 17, 46, 55, 68, 100, 144, 214 & 0.83 & 0.150 & one extra change point \\
            & 0, 46, 55, 68, 144, 214 & 0.95 & 0.079 & one missing change point \\
            & 0, 17, 46, 55, 144, 214 & 0.89 & 0.061 & one missing change point \\
            & 0, 50, 100, 150, 214 & 0.61 & 0.150 & random segmentation \\
            & 0, 214 & 0.00 & 0.327 & no segmentation \\
        \end{tabular}
    }
    \end{small}
\end{table}
\begin{table}[h]

    \caption{
        \label{tab:main_results_hausdorff_median}
        Median relative Hausdorff distance over 500 simulations (in percent), with median absolute deviations in parentheses. Optimal scores are marked in bold.
    }
    \makebox[\textwidth][c]{
        \begin{tabular}{lrrrrr}
         &        CIM &         CIC &   Dirichlet &        iris &        glass \\
        \midrule
        change in mean &  \textbf{0.0} (0.1) &  33.3 (0.0)              &  48.0 (0.0)               &   1.3 (4.9)           &   16.8 (1.5) \\
        changekNN      &  \textbf{0.0} (1.0) &  33.3 (0.6)              &   6.6 (7.7)               &   \textbf{0.0} (1.5)  &    7.9 (3.4) \\
        ECP            &  \textbf{0.0} (1.0) &  32.0 (9.3)              &   5.2 (2.8)               &   \textbf{0.0} (1.2)  &  31.1 (13.2) \\
        KCP            &  \textbf{0.0} (0.3) &  \textbf{0.7} (2.4)      &   5.0 (2.5)               &   \textbf{0.0} (0.8)  &   7.9 (10.1) \\
        MultiRank      &  \textbf{0.0} (0.0) &  33.3 (0.3)              &   3.0 (1.3)               &  15.3 (4.0)           &   15.9 (2.1) \\
        MNWBS        &  \textbf{0.0} (0.7) & 15.3 (12.5)              &   5.9 (0.9)               &  \textbf{0.0} (0.5)   &    8.9 (5.9) \\
        \texttt{changeforest} &  \textbf{0.0} (1.8) & 1.3 (3.7)         &   \textbf{0.2} (1.4)      &   \textbf{0.0} (2.1)  &    \textbf{2.8} (3.7) \\
        \end{tabular}
    }
    \newline
    \vspace{0.2cm}
    \newline
    \makebox[\textwidth][c]{
        \begin{tabular}{lrrrrr}
         & breast cancer &    abalone &       wine &  dry beans & average \\
        \midrule
        change in mean &            16.7 (1.2) &                      5.0 (1.8) &              0.4 (1.4) &             1.3 (1.3) &    13.7 \\
        changekNN      &             \textbf{0.0} (7.5) &            5.0 (3.5) &              1.5 (1.7) &             \textbf{0.0} (0.6) &     6.0 \\
        ECP            &             \textbf{0.0} (0.8) &            6.2 (2.1) &              0.3 (1.7) &             \textbf{0.0} (0.1) &     8.1 \\
        KCP            &             \textbf{0.0} (0.4) &            4.2 (1.6) &              0.2 (1.4) &             \textbf{0.0} (0.0) &     2.0 \\
        MultiRank      &             \textbf{0.0} (7.3) &            9.1 (4.4) &              3.0 (1.1) &                                &    10.0 \\
        MNWBS      &               \textbf{0.0} (0.7) &                      &                        &                                &    5.0 \\
        \texttt{changeforest}  &     \textbf{0.0} (2.1) &   \textbf{3.1} (1.9) &     \textbf{0.0} (0.9) &             \textbf{0.0} (0.6) &     \textbf{0.9} \\
        \end{tabular}
    }
\end{table}
\begin{table}[h]
    \caption{
        \label{tab:main_results_n_cpts}
        The average number of change points estimated over 500 simulations, with standard deviations in parentheses.
    }
    \makebox[\textwidth][c]{
        \begin{tabular}{lrrrrr}
             &        CIM &        CIC &   Dirichlet &        iris &       glass \\
            \midrule
            ground truth   &  2           &       2      &        10     &         2     &     5  \\
            \midrule
            change in mean &  2.00 (0.00) &  0.00 (0.00) &   0.00 (0.00) &   2.99 (1.40) &  22.06 (7.19) \\
            changekNN      &  2.12 (0.41) &  0.58 (1.24) &   8.01 (2.68) &   2.11 (0.35) &   4.16 (0.97) \\
            ECP            &  2.07 (0.34) &  1.21 (1.05) &   7.70 (1.07) &   2.09 (0.36) &   2.58 (1.10) \\
            KCP            &  2.03 (0.18) &  2.39 (0.66) &   8.39 (1.16) &   2.05 (0.31) &   3.44 (1.12) \\
            MultiRank      &  2.00 (0.00) &  0.41 (3.72) &   9.18 (0.46) &  13.31 (8.00) &  21.74 (3.48) \\
            MNWBS         &  2.14 (1.05) &  1.81 (1.25) &   21.96 (4.82) &  2.05 (0.60) &   5.26 (3.05) \\
            \texttt{changeforest}   &  2.13 (0.38) &  2.38 (0.67) &  10.54 (0.77) &   2.15 (0.42) &   4.93 (0.78) \\
        \end{tabular}
    }
    \newline
    \vspace*{0.2cm}
    \newline

    \makebox[\textwidth][c]{
        \begin{tabular}{lrrrr}
            data set & breast cancer &     abalone &       wine &  dry beans \\
            \midrule
            ground truth   &  1              &     14        &      4     &         6     \\
            \midrule
            change in mean &    16.62 (3.94) &  11.07 (1.38) &  3.65 (0.52) &  7.04 (0.98) \\
            changekNN      &     1.49 (0.90) &  11.74 (1.76) &  4.52 (1.00) &  6.21 (0.49) \\
            ECP            &     1.07 (0.36) &   9.73 (1.10) &  3.71 (0.57) &  6.05 (0.25) \\
            KCP            &     1.01 (0.14) &  10.62 (1.07) &  3.60 (0.49) &  6.00 (0.04) \\
            MultiRank      &    8.28 (10.91) &  11.77 (1.93) &  3.00 (0.71) &          \\
            MNWBS        &      1.05 (0.57) &             &                &        \\
            \texttt{changeforest}   &     1.08 (0.34) &  12.09 (1.37) &  4.21 (0.49) &  6.23 (0.47) \\
        \end{tabular}
    }
\end{table}

\begin{table}[h]
\caption{
      \label{tab:tuning_parameters}
      Average adjusted Rand indices of \texttt{changeforest} for different tuning parameters over 500 simulations, with standard deviations in parentheses.
      Optimal scores and those within two standard deviations are marked in bold.
      When necessary, mtry, the number of features evaluated at each split, was rounded down to the next integer.
}
\makebox[\textwidth][c]{
      \begin{tabular}{llllllll}
      trees & max. & mtry     &           CIM         &            CIC        &      Dirichlet        &           iris           &          glass \\
            & depth      &          &                       &                       &                       &                          &          \\
            \midrule
        20  & 2         & 1         &           0.99 (0.03) &           0.13 (0.29) &           0.97 (0.02) &           0.98 (0.04) &           0.89 (0.08) \\
            &           & $\sqrt d$ &           0.99 (0.03) &           0.19 (0.34) &           0.98 (0.02) &           0.98 (0.04) &           0.90 (0.08) \\
            &           & $d$       &           0.99 (0.03) &           0.19 (0.34) &           0.98 (0.02) &           0.99 (0.03) &           0.91 (0.08) \\
            & 8         & 1         &           \textbf{0.99} (0.02) &  0.29 (0.41) &           0.98 (0.01) &           \textbf{0.99} (0.02) &  0.89 (0.08) \\
            &           & $\sqrt d$ &           \textbf{0.99} (0.02) &  0.27 (0.40) &           \textbf{0.99} (0.01) &  \textbf{0.99} (0.02) &  0.89 (0.10) \\
            &           & $d$       &           0.99 (0.05) &           0.21 (0.37) &           \textbf{0.99} (0.01) &  \textbf{0.99} (0.02) &  0.87 (0.13) \\
            & $\infty$  & 1         &           0.87 (0.27) &           0.00 (0.00) &           0.98 (0.02) &           \textbf{0.99} (0.02) &  0.86 (0.15) \\
            &           & $\sqrt d$ &           0.85 (0.30) &           0.00 (0.00) &           \textbf{0.99} (0.01) &  \textbf{0.99} (0.02) &  0.87 (0.12) \\
            &           & $d$       &           0.78 (0.36) &           0.00 (0.00) &           \textbf{0.99} (0.01) &  \textbf{0.99} (0.02) &  0.85 (0.16) \\
      100   & 2         & 1         &           0.98 (0.04) &           0.50 (0.44) &           0.98 (0.02) &           0.98 (0.04) &           0.90 (0.08) \\
            &           & $\sqrt d$ &           0.98 (0.04) &           0.52 (0.44) &           0.98 (0.02) &           0.98 (0.04) &           0.91 (0.07) \\
            &           & $d$       &           0.98 (0.04) &           0.50 (0.43) &           0.98 (0.02) &           0.98 (0.04) &           \textbf{0.92} (0.07) \\
            & 8         & 1         &           0.99 (0.03) &           0.93 (0.13) &           0.99 (0.02) &           0.98 (0.04) &           \textbf{0.92} (0.07) \\
            &           & $\sqrt d$ &           0.99 (0.03) &           0.92 (0.12) &           0.99 (0.02) &           0.98 (0.04) &           \textbf{0.92} (0.07) \\
            &           & $d$       &           \textbf{0.99} (0.03) &  0.89 (0.20) &           0.98 (0.02) &           0.99 (0.03) &           0.92 (0.08) \\
            &  $\infty$ & 1         &           0.98 (0.04) &           0.88 (0.22) &           0.99 (0.02) &           0.98 (0.04) &           0.92 (0.08) \\
            &           & $\sqrt d$ &           0.99 (0.04) &           0.88 (0.22) &           \textbf{0.99} (0.02) &  0.98 (0.04) &           \textbf{0.92} (0.07) \\
            &           & $d$       &           0.99 (0.03) &           0.82 (0.31) &           0.98 (0.02) &           0.99 (0.04) &           0.92 (0.08) \\
      500   & 2         & 1         &           0.98 (0.05) &           0.53 (0.43) &           0.98 (0.02) &           0.98 (0.05) &           0.91 (0.07) \\
            &           & $\sqrt d$ &           0.98 (0.04) &           0.57 (0.43) &           0.98 (0.03) &           0.98 (0.04) &           0.92 (0.07) \\
            &           & $d$       &           0.98 (0.04) &           0.52 (0.43) &           0.97 (0.02) &           0.98 (0.04) &           0.92 (0.07) \\
            & 8         & 1         &           0.98 (0.04) &           \textbf{0.95} (0.05) &  0.99 (0.02) &           0.98 (0.05) &           \textbf{0.92} (0.07) \\
            &           & $\sqrt d$ &           0.98 (0.04) &           0.94 (0.07) &           0.99 (0.02) &           0.98 (0.04) &           \textbf{0.92} (0.07) \\
            &           & $d$       &           0.99 (0.04) &           0.93 (0.11) &           0.98 (0.02) &           0.99 (0.04) &           \textbf{0.92} (0.08) \\
            &  $\infty$ & 1         &           0.98 (0.04) &           0.94 (0.08) &           \textbf{0.99} (0.02) &  0.98 (0.04) &           \textbf{0.92} (0.07) \\
            &           & $\sqrt d$ &           0.98 (0.04) &           \textbf{0.94} (0.07) &  0.99 (0.02) &           0.98 (0.04) &           \textbf{0.92} (0.07) \\
            &           & $d$       &           0.98 (0.04) &           0.93 (0.11) &           0.98 (0.02) &           0.99 (0.04) &           0.92 (0.08) \\
      \end{tabular}
}
\vspace*{0.5cm}
Continued on the next page.
\end{table}

\begin{table}[h]
      \captionsetup{labelformat=empty}
      \caption{Table \ref{tab:tuning_parameters} (Continued):
      Average adjusted Rand indices of \texttt{changeforest} for different tuning parameters over 500 simulations, with standard deviations in parentheses.
      Optimal scores and those within two standard deviations are marked in bold.
      When necessary, mtry, the number of features evaluated at each split, was rounded down to the next integer.}
      \addtocounter{table}{-1}
      \makebox[\textwidth][c]{
            \begin{tabular}{llllllll}
            trees & max. & mtry     & breast            &        abalone             &   wine                   &      dry          &   average \\
                  & depth      &     &  cancer                     &                       &                       & beans                         &          \\ 
                  \midrule
            20  & 2           & 1         &           0.99 (0.05) &           0.87 (0.08) &           0.98 (0.02) &           1.00 (0.01) &           0.87 (0.11) \\
                  &           & $\sqrt d$ &           0.99 (0.04) &           0.91 (0.06) &           0.99 (0.01) &           1.00 (0.01) &           0.88 (0.12) \\
                  &           & $d$       &           0.99 (0.05) &           0.92 (0.05) &           0.99 (0.02) &           1.00 (0.01) &           0.88 (0.12) \\
                  & 8         & 1         &           1.00 (0.03) &           0.84 (0.11) &           1.00 (0.01) &  \textbf{1.00} (0.00) &           0.89 (0.14) \\
                  &           & $\sqrt d$ &           0.99 (0.04) &           0.85 (0.11) &           1.00 (0.01) &  \textbf{1.00} (0.00) &           0.89 (0.14) \\
                  &           & $d$       &           0.99 (0.04) &           0.81 (0.13) &  \textbf{1.00} (0.00) &  \textbf{1.00} (0.00) &           0.87 (0.14) \\
                  & $\infty$  & 1         &  \textbf{1.00} (0.01) &           0.08 (0.14) &           0.99 (0.01) &  \textbf{1.00} (0.00) &           0.75 (0.11) \\
                  &           & $\sqrt d$ &  \textbf{1.00} (0.01) &           0.09 (0.14) &           0.99 (0.01) &  \textbf{1.00} (0.00) &           0.75 (0.12) \\
                  &           & $d$       &  \textbf{1.00} (0.01) &           0.07 (0.12) &           0.99 (0.01) &  \textbf{1.00} (0.00) &           0.74 (0.14) \\
            100   & 2         & 1         &           0.98 (0.09) &           0.91 (0.05) &           0.98 (0.04) &           0.99 (0.02) &           0.91 (0.15) \\
                  &           & $\sqrt d$ &           0.98 (0.08) &           0.93 (0.04) &           0.98 (0.04) &           0.99 (0.02) &           0.92 (0.15) \\
                  &           & $d$       &           0.97 (0.08) &           0.94 (0.04) &           0.98 (0.04) &           0.99 (0.03) &           0.92 (0.15) \\
                  & 8         & 1         &           0.98 (0.08) &           0.93 (0.05) &           0.99 (0.03) &           0.99 (0.01) &  \textbf{0.97} (0.06) \\
                  &           & $\sqrt d$ &           0.98 (0.07) &           0.93 (0.05) &           0.99 (0.02) &           1.00 (0.01) &  \textbf{0.97} (0.06) \\
                  &           & $d$       &           0.98 (0.07) &           0.93 (0.06) &           0.99 (0.02) &           1.00 (0.01) &           0.96 (0.08) \\
                  & $\infty$  & 1         &           1.00 (0.02) &           0.88 (0.10) &  \textbf{1.00} (0.00) &  \textbf{1.00} (0.01) &           0.96 (0.09) \\
                  &           & $\sqrt d$ &           0.99 (0.03) &           0.90 (0.07) &  \textbf{1.00} (0.01) &           1.00 (0.01) &           0.96 (0.08) \\
                  &           & $d$       &           1.00 (0.02) &           0.89 (0.08) &  \textbf{1.00} (0.00) &           0.99 (0.02) &           0.95 (0.11) \\
            500   & 2         & 1         &           0.98 (0.09) &           0.92 (0.05) &           0.97 (0.06) &           0.98 (0.03) &           0.91 (0.15) \\
                  &           & $\sqrt d$ &           0.97 (0.09) &           0.93 (0.04) &           0.97 (0.06) &           0.98 (0.03) &           0.92 (0.15) \\
                  &           & $d$       &           0.97 (0.09) &  \textbf{0.94} (0.04) &           0.97 (0.06) &           0.98 (0.03) &           0.92 (0.15) \\
                  & 8         & 1         &           0.97 (0.10) &  \textbf{0.93} (0.05) &           0.98 (0.04) &           0.99 (0.02) &  \textbf{0.97} (0.05) \\
                  &           & $\sqrt d$ &           0.98 (0.08) &  \textbf{0.94} (0.04) &           0.99 (0.04) &           0.99 (0.02) &  \textbf{0.97} (0.05) \\
                  &           & $d$       &           0.98 (0.08) &           0.93 (0.05) &           0.99 (0.02) &           1.00 (0.01) &  \textbf{0.97} (0.06) \\
                  & $\infty$  & 1         &           0.99 (0.03) &           0.91 (0.06) &  \textbf{1.00} (0.01) &           1.00 (0.01) &  \textbf{0.97} (0.05) \\
                  &           & $\sqrt d$ &           0.99 (0.03) &           0.93 (0.05) &  \textbf{1.00} (0.01) &           0.99 (0.02) &  \textbf{0.97} (0.04) \\
                  &           & $d$       &           0.99 (0.03) &           0.92 (0.06) &           1.00 (0.01) &           0.99 (0.02) &  \textbf{0.97} (0.05) \\
            \end{tabular}
      }
      \vspace*{-2mm}
\end{table}    

\begin{table}[h]
    \caption{
        \label{tab:tuning_parameters_time}
        Average computational times of \texttt{changeforest} in seconds on 8 Intel Xeon 2.3 GHz cores.
        When necessary, mtry, the number of features evaluated at each split, was rounded down to the next integer.
    }
    \makebox[\textwidth][c]{
    \begin{tabular}{llllllllllll}
    trees & max. & mtry &   CIM &   CIC & Dirichlet &  iris & glass &  wine & breast & abalone & dry \\
          & depth&      &       &       &           &       &       &       & cancer &         & beans \\
    \midrule
        20  & 2 & 1 &  0.03 &  0.01 &      0.10 &  0.01 &  0.02 &          0.02 &    0.23 &  0.24 &      0.64 \\
        &      & $\sqrt{d}$ &  0.02 &  0.01 &      0.11 &  0.01 &  0.02 &          0.02 &    0.24 &  0.25 &      0.68 \\
        &      & $d$ &  0.03 &  0.01 &      0.14 &  0.01 &  0.02 &          0.02 &    0.27 &  0.27 &      0.86 \\
            & 8 & 1 &  0.03 &  0.02 &      0.16 &  0.01 &  0.02 &          0.03 &    0.32 &  0.34 &      0.84 \\
            &      & $\sqrt{d}$ &  0.03 &  0.02 &      0.16 &  0.01 &  0.03 &          0.03 &    0.35 &  0.36 &      0.97 \\
            &      & $d$ &  0.04 &  0.03 &      0.21 &  0.01 &  0.03 &          0.03 &    0.42 &  0.45 &       1.5 \\
            & $\infty$ & 1 &  0.04 &  0.02 &      0.18 &  0.01 &  0.03 &          0.03 &    0.23 &  0.57 &       1.2 \\
            &      & $\sqrt{d}$ &  0.04 &  0.02 &      0.17 &  0.01 &  0.03 &          0.03 &    0.27 &  0.57 &       1.4 \\
            &      & $d$ &  0.05 &  0.02 &      0.22 &  0.01 &  0.03 &          0.03 &    0.31 &  0.72 &       2.1 \\
        100 & 2 & 1 &  0.04 &  0.03 &      0.22 &  0.02 &  0.03 &          0.04 &    0.43 &  0.45 &       1.3 \\
        &      & $\sqrt{d}$ &  0.04 &  0.03 &      0.24 &  0.02 &  0.04 &          0.04 &    0.46 &  0.48 &       1.4 \\
        &      & $d$ &  0.05 &  0.04 &      0.33 &  0.02 &  0.04 &          0.04 &    0.56 &  0.56 &       2.1 \\
            & 8 & 1 &  0.07 &  0.08 &      0.43 &  0.02 &  0.06 &          0.07 &    0.77 &  0.79 &       2.1 \\
            &      & $\sqrt{d}$ &  0.08 &  0.09 &      0.43 &  0.02 &  0.06 &          0.07 &    0.85 &  0.87 &       2.6 \\
            &      & $d$ &  0.10 &  0.10 &      0.57 &  0.02 &  0.07 &          0.07 &     1.1 &   1.2 &       4.7 \\
            & $\infty$ & 1 &  0.09 &  0.10 &      0.50 &  0.03 &  0.06 &          0.09 &     1.3 &   1.6 &       3.6 \\
            &      & $\sqrt{d}$ &  0.10 &  0.11 &      0.46 &  0.02 &  0.06 &          0.09 &     1.4 &   1.6 &       4.2 \\
            &      & $d$ &  0.11 &  0.12 &      0.60 &  0.03 &  0.07 &          0.09 &     1.7 &   2.2 &       7.4 \\
        500 & 2 & 1 &  0.11 &  0.08 &      0.67 &  0.04 &  0.10 &          0.10 &     1.2 &   1.4 &       4.3 \\
        &      & $\sqrt{d}$ &  0.12 &  0.09 &      0.75 &  0.04 &  0.10 &          0.11 &     1.3 &   1.5 &       5.0 \\
        &      & $d$ &  0.14 &  0.11 &       1.1 &  0.05 &  0.12 &          0.12 &     1.7 &   1.9 &       8.2 \\
            & 8 & 1 &  0.25 &  0.27 &       1.6 &  0.08 &  0.20 &          0.22 &     2.7 &   3.0 &       8.4 \\
            &      & $\sqrt{d}$ &  0.27 &  0.30 &       1.6 &  0.08 &  0.20 &          0.22 &     3.1 &   3.4 &        11 \\
            &      & $d$ &  0.33 &  0.37 &       2.2 &  0.08 &  0.22 &          0.25 &     4.2 &   5.0 &        21 \\
            & $\infty$ & 1 &  0.32 &  0.36 &       1.9 &  0.08 &  0.22 &          0.30 &     4.9 &   6.8 &        15 \\
            &      & $\sqrt{d}$ &  0.33 &  0.38 &       1.7 &  0.08 &  0.21 &          0.29 &     5.3 &   7.0 &        18 \\
            &      & $d$ &  0.38 &  0.44 &       2.3 &  0.08 &  0.22 &          0.31 &     6.7 &   9.7 &        33 \\
    \end{tabular}
    }
\end{table}

\begin{table}[h]
    \caption{
        \label{tab:tuning_parameters_kcp}
        Average adjusted Rand indices of KCP \citep{arlot2019kernel} for different tuning parameters over 500 simulations, with standard deviations in parentheses.
        Optimal scores and those within two standard deviations are marked in bold.
        Note that for all nonparametric simulation setups, the covariates have been normalized by the median absolute deviation of absolute consecutive differences, such that the reported bandwidths are relative to the scale of the covariates.
        For the bandwidth \emph{median}, the bandwidth was chosen according to the median heuristic \citep{garreau2017large}.
        Here, differently from the normalization based on consecutive differences, the median is taken over all possible pairwise distances and thus mostly contains differences across different distributions.
        As a consequence, the bandwidths selected by the median heuristic tend to be large (greater than 1).
        The oracle is the best performing hyperparameter combination.
        We used a Gaussian kernel with bandwidth 0.1 in the main simulations.
    }
    \makebox[\textwidth][c]{
      \begin{tabular}{llllllll}
            kernel   & bandwidth & CIM       &   CIC        &      Dirichlet      &           iris          &          glass \\
            \midrule
            cosine   &        &  0.80 (0.17) &  0.36 (0.25) &         0.86 (0.08)           &  0.03 (0.00) &  0.03 (0.00) \\
            linear   &        &  0.99 (0.03) &  0.34 (0.28) &         \textbf{0.87} (0.08)  &  0.85 (0.16) &  0.37 (0.15) \\
            Gaussian & 0.025  &  0.99 (0.03) &  0.38 (0.36) &         0.80 (0.11)           &  0.94 (0.11) &  \textbf{0.83} (0.16) \\
            Gaussian & 0.05   &  1.00 (0.02) &  0.73 (0.33) &         \textbf{0.87} (0.08)  &  0.98 (0.06) &  \textbf{0.83} (0.17) \\
            Gaussian & 0.1    &  1.00 (0.01) &  0.95 (0.06) &         \textbf{0.87} (0.08)  &  0.99 (0.02) &  0.74 (0.23) \\
            Gaussian & 0.2    &  1.00 (0.00) &  0.98 (0.03) &         \textbf{0.87} (0.08)  &  \textbf{1.00} (0.01) &  0.59 (0.26) \\
            Gaussian & 0.4    &  1.00 (0.00) &  \textbf{0.98} (0.02) & \textbf{0.87} (0.08) &  0.99 (0.02) &  0.46 (0.28) \\
            Gaussian & 0.8    &  \textbf{1.00} (0.00) &  0.97 (0.09) & \textbf{0.87} (0.08) &  0.83 (0.22) &  0.15 (0.25) \\
            Gaussian & median &  0.00 (0.00) &  0.00 (0.00) &         \textbf{0.87} (0.08)  &  0.00 (0.00) &  0.00 (0.00) \\
            \bottomrule
            oracle      &     &  1.00 (0.00) &  0.98 (0.02) &  0.87 (0.08) &  1.00 (0.01) &  0.83 (0.17) \\
      \end{tabular}
    }
    \newline
    \vspace*{0.2cm}
    \newline
    \makebox[\textwidth][c]{
        \begin{tabular}{lllllll}
            kernel      & bandwidth &  breast cancer          &        abalone          &           wine          &      dry beans &        average \\
            \midrule
            cosine      &           &   0.04 (0.01) &  0.39 (0.04) &  0.16 (0.02) &  0.79 (0.08) &  0.38 (0.11) \\
            linear      &           &   0.97 (0.05) &  0.84 (0.06) &  0.63 (0.16) &  0.96 (0.02) &  0.76 (0.14) \\
            Gaussian    &    0.025  &   \textbf{1.00} (0.00) &  0.89 (0.05) &  0.99 (0.02) &  \textbf{1.00} (0.00) &  0.87 (0.14) \\
            Gaussian    &    0.05   &   \textbf{1.00} (0.00) &  \textbf{0.90} (0.05) &  \textbf{0.99} (0.01) &  \textbf{1.00} (0.00) &  0.92 (0.13) \\
            Gaussian    &    0.1    &   1.00 (0.03) &  \textbf{0.92} (0.04) &  \textbf{0.99} (0.01) &  \textbf{1.00} (0.00) &  \textbf{0.94} (0.09) \\
            Gaussian    &    0.2    &   0.99 (0.06) &  0.92 (0.04) &  \textbf{0.99} (0.01) &  \textbf{1.00} (0.00) &  0.93 (0.09) \\
            Gaussian    &    0.4    &   0.99 (0.05) &  0.88 (0.08) &  0.97 (0.03) &  \textbf{1.00} (0.00) &  0.91 (0.10) \\
            Gaussian    &    0.8    &   0.99 (0.05) &  0.82 (0.10) &  0.92 (0.03) &  0.97 (0.04) &  0.84 (0.12) \\
            Gaussian    &    median &   0.98 (0.06) &  0.61 (0.13) &  0.00 (0.03) &  0.00 (0.00) &  0.27 (0.06) \\
            \bottomrule
            oracle      &       &  1.00 (0.00)          &  0.92 (0.04)          &    0.99 (0.01)        &  1.00 (0.00)          &  0.95 (0.07) \\
        \end{tabular}
    }
\end{table}

\end{document}